\documentclass [a4paper,12pt,oneside]{book}

\usepackage{dsfont}
\usepackage {braket}
\usepackage[top=3.5cm, bottom=2.5cm, left=3cm, right=3cm]{geometry}
\usepackage [english]{babel}
\usepackage [latin1]{inputenc}
\usepackage{slashed}
\usepackage {amsmath}
\usepackage{amsthm}
\renewcommand*{\theequation}{%
  \ifnum\value{section}=0 %
    \thechapter
  \else
    \thesection
  \fi
  .\arabic{equation}%
}

\usepackage{amssymb}
\usepackage {graphicx,color}
\usepackage{enumerate}
\usepackage {multirow}
\usepackage{graphicx}
\usepackage{subcaption}
\usepackage {booktabs}
\usepackage {fancyhdr}
\usepackage[pdftex]{hyperref}
\usepackage{mathtools}
\usepackage{wasysym}
\usepackage{pdfpages}
\usepackage{mathrsfs}
\graphicspath{ {figures/} }
\usepackage{textcomp}
\newcommand{\Tau}{\mathcal{T}}
\newcommand{\del}{\partial}
\numberwithin{equation}{section}

\newcommand\norm[1]{\left\lVert#1\right\rVert}
\newtheorem{theorem}{Theorem}[chapter]
\newtheorem{mydef}{Definition}[chapter]

\newtheorem{lemma}[theorem]{Lemma}
\usepackage{tikz}

\begin{document}
\pagenumbering{gobble}
\thispagestyle{empty}
\begin{center}
	{\fontsize{16}{16} \selectfont Universidade de S\~ao Paulo \\}
	\vspace{0.1cm}
	{\fontsize{16}{16} \selectfont Instituto de F\'{i}sica}
    \vspace{3.3cm}

	{\fontsize{22}{22}\selectfont RG Flows e Sistemas Din\^{a}micos\par}
    \vspace{2cm}
    

    {\fontsize{18}{18}\selectfont Caio Luiz Tiedt\par}

    \vspace{2cm}
\end{center}
   
\leftskip 6cm
\begin{flushright}	
\leftskip 6cm
Orientador: Prof. Dr. Diego Trancanelli 
\leftskip 6cm
\end{flushright}	

    \vspace{0.8cm} 
    

\par
\leftskip 6cm
\noindent {Disserta\c{c}\~{a}o de mestrado apresentada ao Instituto\\ de F\'{i}sica da Universidade de S\~{a}o Paulo, como re-\\quisito parcial para a obten\c{c}\~{a}o do t\'{i}tulo de Mestre(a)\\ em Ci\^{e}ncias.}
\par
\leftskip 0cm
\vskip 2cm


\noindent Banca Examinadora: \\
\noindent Prof. Dr. Diego Trancanelli - Orientador (IF-USP)\\
Prof. Dr. Dmitry Melnikov (IIP-UFRN) \\
Prof. Dr. Horatiu Nastase (IFT-UNESP) \\
\vspace{2.8cm}  

\begin{center}
S\~{a}o Paulo \\ 2019
\end{center}

\newpage
\pagenumbering{gobble}
\includepdf[pages={1}]{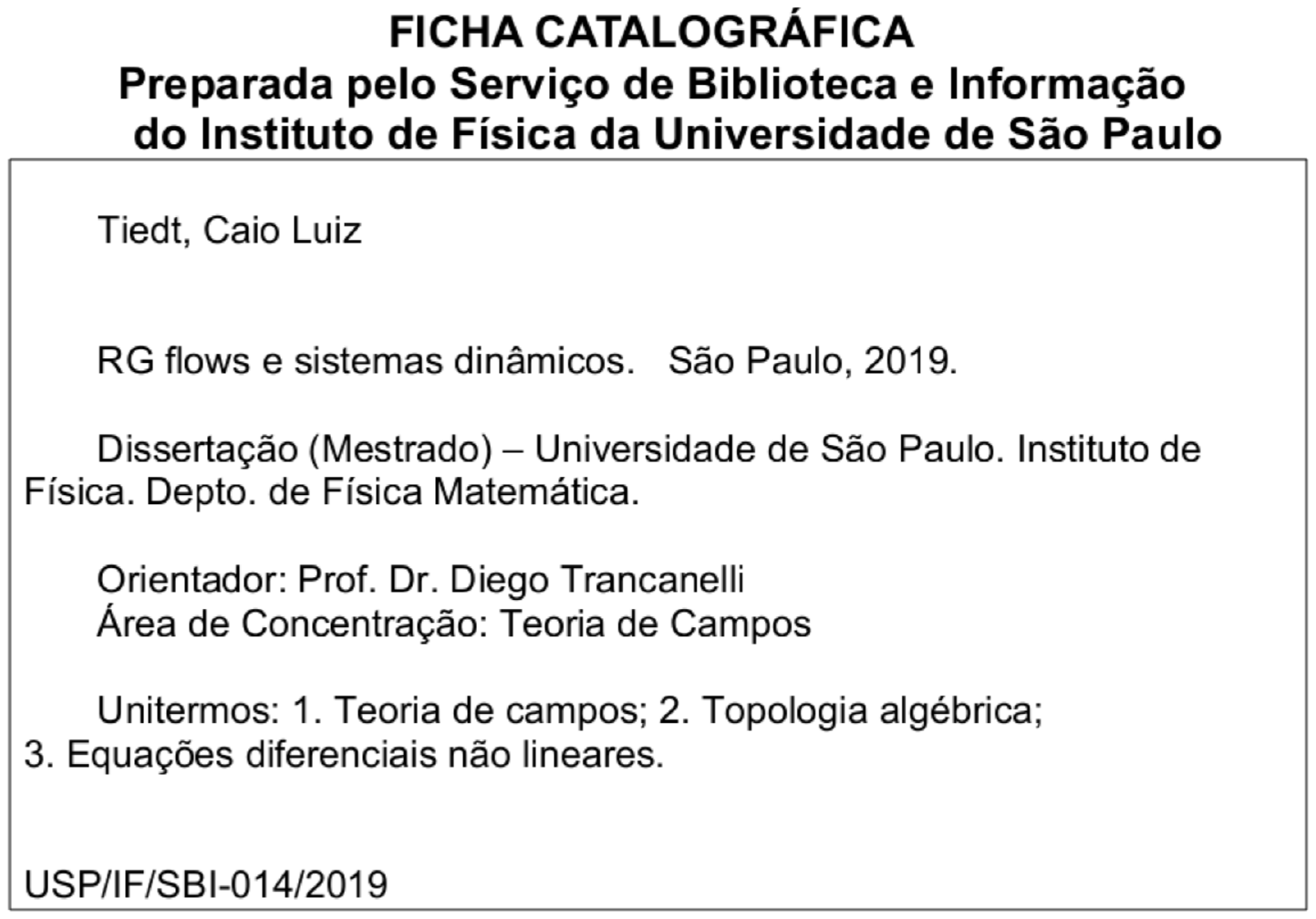}
\clearpage

\newpage
\pagenumbering{gobble}
\thispagestyle{empty}
\begin{center}
	{\fontsize{16}{16} \selectfont University of S\~ao Paulo \\}
	\vspace{0.1cm}
	{\fontsize{16}{16} \selectfont Physics Institute}
    \vspace{3.3cm}

	{\fontsize{22}{22}\selectfont RG Flows and Dynamical Systems \par}
    \vspace{2cm}
    

    {\fontsize{18}{18}\selectfont Caio Luiz Tiedt\par}

    \vspace{2cm}
\end{center}
   
\leftskip 6cm
\begin{flushright}	
\leftskip 6cm
Supervisor: Prof. Dr. Diego Trancanelli  
\leftskip 6cm
\end{flushright}	

    \vspace{0.8cm} 
    

\par
\leftskip 6cm
\noindent {Dissertation submitted to the Physics Institute of the University of S\^{a}o Paulo in partial fulfillment of the requirements for the degree of Master of Science.}
\par
\leftskip 0cm
\vskip 2cm


\noindent Examining Committee: \\
\noindent Prof. Dr. Diego Trancanelli - Supervisor (IF-USP)\\
Prof. Dr. Dmitry Melnikov (IIP-UFRN) \\
Prof. Dr. Horatiu Nastase (IFT-UNESP) \\
\vspace{2.8cm}

\begin{center}
S\~{a}o Paulo \\ 2019
\end{center}

\newpage
\pagenumbering{gobble}
\thispagestyle{empty}
\mbox{}

\newpage
 \pagenumbering{arabic}
 \thispagestyle{empty}
 	\vspace*{\fill}
 	{ \raggedleft

	\textit{To Marielle Franco.}

	~
	}

	\newpage
	\pagenumbering{arabic}
	\thispagestyle{empty}
	
	\section*{Acknowledgements}
	
    My thanks go to Prof. Diego Trancanelli for guiding me through the world of field theory as well as for believing in this project. I thank my colleague Gabriel Nagaoka for all kinds of support, and also Daniel Lombelo Teixeira for excellent advise. I thank Prof. Sergei Gukov, whom I only met once, but he inspired me to go forward. Finally, I thank CNPq for the scholarship that was given to me.

\newpage

 	\newpage
 	\thispagestyle{empty}
 	\vspace*{\fill}
 	{ \raggedleft

	\textit{``Utopia [...]\\
	ella est\'a en el horizonte. Me acerco dos pasos, ella se aleja dos pasos.  \\
	Camino diez pasos y el horizonte se corre diez pasos m\'as all\'a.\\
	Por mucho que yo camine, nunca la alcanzar\'e.  \\
	Para que sirve la utopia? Para eso sirve: para caminar." \\
	Eduardo  Galeano,  ``Las  palabras  andantes"}

	~
	}
 	
 	\newpage
 	\chapter*{Resumo}
 	\thispagestyle{empty}
 	
    No contexto de Renormaliza\c{c}\~ao Wilsoniana, os fluxos do grupo de renormaliza\c{c}\~ao (RG flows) s\~ao um conjunto de equa\c{c}\~oes diferenciais que define como as constantes de acoplamento de uma teoria dependem de uma escala de energia. o conte\'udo destes \'e semelhante a como sistemas termodin\^amicos est\~ao relacionados com a temperatura. Neste sentindo, \'e natural olhar para estruturas nos fluxos que demonstram um comportamento termodin\^amico. A teoria matem\'atica para estudar estas equa\c{c}\~oes \'e chamada de sistemas din\^amicos e aplica\c{c}\~oes desta t\^em sido usadas no estudo de RG flows. Como exemplo o teorema-C de Zamolodchikov e os equivalentes teoremas em dimens\~oes maiores mostram que existe uma fun\c{c}\~ao monotonicamente decrescente ao longo do fluxo e \'e uma propriedade que se assemelha \`a segunda lei da termodin\^amica, est\~ao relacionadas com a fun\c{c}\~ao de Lyapunov no contexto de sistemas din\^amicos e podem ser usadas para excluir a possibilidade de comportamentos assint\'oticos ex\'oticos, como fluxos peri\'odicos ou ciclos limites. Estudamos a teoria de bifurca\c{c}\~ao e a teoria de \'indice, que foram propostas como sendo \'uteis no estudo de RG flows: a primeira pode ser usada para explicar constantes cruzando pela marginalidade e a segunda para extrair informa\c{c}\~ao global do espa\c{c}o em que os fluxos vivem. Nesta disserta\c{c}\~ao, tamb\'em olhamos para aplica\c{c}\~oes em RG flows hologr\'aficos e tentamos buscar rela\c{c}\~oes entre as estruturas em teorias hologr\'aficas e as suas duais teorias de campos.

	Palavras-Chave: ``Grupo de renormaliz\c{c}\~ao", ``Sistemas Din\^amicos", ``Fluxos Peri\'odicos", ``Teorema-C", ``Teoria de Bifurca\c{c}\~ao", ``Homologia de Morse", ``Holografia".

 	\newpage 	
 	\chapter*{Abstract}
 	 \thispagestyle{empty}
 	In the context of Wilsonian Renormalization, renormalization group (RG) flows are a set of differential equations that defines how the coupling constants of a theory depend on an energy scale. These equations closely resemble thermodynamical equations and how thermodynamical systems are related to temperature. In this sense, it is natural to look for structures in the flows that show a thermodynamics-like behaviour. The mathematical theory to study these equations is called Dynamical Systems, and applications of that have been used to study RG flows. For example, the classical Zamolodchikov's C-Theorem and its higher-dimensional counterparts, that show that there is a monotonically decreasing function along the flow and it is a property that resembles the second-law of thermodynamics, is related to the Lyapunov function in the context of Dynamical Systems. It can be used to rule out exotic asymptotic behaviours like periodic flows (also known as limit cycles). We also study bifurcation theory and index theories, which have been proposed to be useful in the study of RG flows, the former can be used to explain couplings crossing through marginality and the latter to extract global information about the space the flows lives in. In this dissertation, we also look for applications in holographic RG flows and we try to see if the structural behaviours in holographic theories are the same as the ones in the dual field theory side.
 	
 Key-Words: ``RG Flows", ``Dynamical Systems", ``Limit Cycles", ``C-Theorem", ``Bifurcation Theory", ``Morse Homology", ``Holography".

 	\newpage
 	\tableofcontents
 	 	\thispagestyle{empty}
	
	\newpage
	\listoffigures
	 \thispagestyle{empty}

	\addtocontents{toc}{\protect\thispagestyle{empty}}
	\addtocontents{lof}{\protect\thispagestyle{empty}}
	\newpage
 
    \pagestyle{fancy}
    \renewcommand{\sectionmark}[1]{\markright{\thesection\ #1}}
    \fancyhf{}
    \renewcommand{\footrulewidth}{0.0pt}
    \renewcommand{\headrulewidth}{0.1pt}
    \lhead{ \bfseries \rightmark}
    \rhead{\bfseries \thepage}
    
\chapter{Introduction}

Renormalization is an important concept in statistical mechanics and in quantum field theory. It was initially developed as a technique to remove arising singularities in quantum field theories studied in high energy physics, like QED. It was only in the 1960's and 1970's that a new outlook at the topic was developed, mainly by Kenneth Wilson \cite{First}, that provided a good physical interpretation for the process of quantum divergences. 

In perturbative field theory, the singularities appears in integrals over the momentum space of the form

\begin{equation}
    \int_p \frac{d^d p}{(2\pi)^d} e^{ipx}p^n  .
\end{equation}
Obviously for $n \geq -1$ the integral diverges. One method of renormalizing the theories that existed before Wilson's contribution was to remove singularities like the one above by adding specific counterterms in the action of a theory that annihilated the divergent terms.

Wilsonian renormalization gave us a much deeper interpretation of the process. It starts by adding a cutoff $\Lambda$ on the momentum.

\begin{equation}
    \int_{|p|<\Lambda} \frac{d^d p}{(2\pi)^d} e^{ipx} F(p)  .
\end{equation}
Of course $\Lambda$ is an artificial parameter and our physical results should not depend on it. Wilson added a second parameter, $b<1$, that serves as a momentum scale on the theory. For $\Lambda>|p|>b\Lambda$ the modes were said to be \textit{slow modes} and for $|p|>b\Lambda$ the modes are called \textit{fast modes}

\begin{equation}\label{idontknowiguess}
    \int_{|p|<\Lambda} \frac{d^d p}{(2\pi)^d}e^{ipx} F(p)  =   \int_{|p|<b\Lambda} \frac{d^d p}{(2\pi)^d}e^{ipx} F_{slow}(p) + \int_{|p|>b\Lambda} \frac{d^d p}{(2\pi)^d}e^{ipx} F_{fast}(p).
\end{equation}

We assume to be able to perform the integral over the fast modes. So the fast modes are going to serve as terms in an effective theory for the slow modes. 

Of course, the parameters on the new effective theory are going to depend on the momentum scale we are using. A similar construction exist in Statistical Mechanics, and that is why Wilsonian renormalization is so important: the scale can be seen as the temperature the theory lives in and the divergences of the unrenormalized theory comes from the fact that the integral is calculating the contributions of all possible energies. The interpretation is that physics depend on the energy scale that is available and that the renormalization process gives us a family of theories that have the same origin but differ since they are being calculated at different temperatures. For historical reasons, this family of theories is called \textit{renormalization group} (RG).

The exact form on how a set of parameters of a theory, called $\lambda_i$, depends on the energy scale, called RG flows, is a set of differential equations of the form

\begin{equation}\label{flowequations}
    \frac{d \lambda_i}{dt} = \beta_i (\lambda).
\end{equation}
Here $t$ is the logarithm of the energy scale and $\beta$ can be calculated through an equation know as the Callan\textendash Symanzik equation\cite{Peskin}. This set of differential equations might not be linear and it might be really hard or even impossible to solve analytically. The mathematical theory that was developed to deal with this problems is called Dynamical Systems, and the idea is to gather the maximum amount of information of the solutions without actually solving the system of equations. 

An important analysis of a dynamical system is to look for the asymptotic behaviour of the solutions. The most simple behaviour is solutions that end up on a specific point of the space. These are known as \textit{fixed points} and in quantum field theory those point are related to a specific set of theories known as Conformal Field Theories (CFT). Another asymptotic behaviour that has been explored is solution that are periodic, known as limit cycles. The meaning behind these solutions in RG flows is something that is still up to debate and a lot of effort has been put in to find conditions that rule out them.

Recently, it has been suggested, specially by Sergei Gukov \cite{GukovBif}, that techniques developed in the context of dynamical systems can be helpful in the analysis of RG flows. In cases of flows calculated using perturbation theory, it is interesting to find topological invariants in the space of the system that are going to be invariant under smooth perturbations of the flow. We also study phase transitions in the space as a control parameter in the theory changes, leading to what is called as bifurcations.

RG flow are also present in holographic theories, a set of theories that tries to relate quantities from quantum field theories to quantities in quantum gravity theories formulated using string theory.

This thesis is organized as follows: In Chapter 2 we will review the renormalization group in statistical mechanics, in order to show the origins of the procedure, and in quantum field theory, where we will show how one can compute the flow equations (\ref{flowequations}) and we will give a construction of the space that the solutions live in.

In Chapter 3, we will review topics on algebraic topology which are going to be needed to define a characteristic of a manifold called Conley index, a number that is going to gives us global information about the RG flows.

Chapter 4 will introduce the theory of dynamical systems, and we will review index theory and bifurcation theory. Also in this chapter we finally make the connection between dynamical systems and RG flows.

In Chapter 5, we will make use of the techniques presented in the previous chapters in examples in quantum field theory. We also give an introduce the concept of RG flows in holographic theories, and we apply bifurcation theory for an example of these theories.

\newpage

\chapter{RG Flows}    	

\section{RG flows in Statistical Mechanics}\label{IsingModelModel}

A natural starting point to discuss RG flows is the Ising model (see for example \cite{Tong}), which is a lattice model of $N$ spins in $d$ dimensions, where each spin can be either $s_i=+1$ or $s_i=-1$.

In the presence of a magnetic field $B$ the energy of the system is given by

\begin{equation}
    E = -B\sum_i s_i - J\sum_{(i,j)}s_is_j.
\end{equation}
The second term is the interaction between spins, and $J$ measures the intensity of the interaction. Different values of the parameters ($B,J$) lead to different behaviours in the system. For example, for $J>0$ minimization of energy tells us that spins will tend to align and such systems are known as ferromagnetic, while for $J<0$ the spins will tend to antialign and these systems are known as anti-ferromagnetic. In this sense, we can say that the possible states for this system span a 2-dimensional space parameterized by ($B,J$).

The partition function of a system is defined by

\begin{equation}
    Z = \sum_{s_i}e^{-\beta E},
\end{equation}
where $\beta = 1/T$, $T$ being the temperature of the system. This function encapsulates a lot of information on the system. For example, the probability of a state $\left\lbrace s_i \right\rbrace$ is given by

\begin{equation}
    p = \frac{e^{-\beta E}}{Z}.
\end{equation}

We can define the magnetization ($\phi$)  as the average spin of the system, which can be calculated by

\begin{equation}
    \phi = \frac{1}{N}\sum_{s_i}\frac{e^{-\beta E}}{Z}\sum_i s_i = \frac{1}{N\beta}\frac{\partial \log Z}{\partial B}.
\end{equation}

Our objective is to calculate the free energy of the system $F$. It can be defined in thermodynamical terms as

\begin{equation}
    F_{thermo} = <E> - TS = -T\log Z.
\end{equation}

We are actually looking for an effective free energy that depends on the magnetization of the system $F(\phi )$. Now we rewrite the partition function in terms of the free energy

\begin{equation}
    Z = \sum_\phi e^{-\beta F(\phi )}.
\end{equation}
In the limit of large $N$, we can transform the sum into an integral

\begin{equation}
    Z = \frac{N}{2}\int_{-1}^1 d\phi e^{-\beta F(m)}.
\end{equation}

So far we've been using a global definition of the magnetization. A generalization of this is the Landau-Ginzbourg model, where we consider $\phi$ to depend on space $\phi (x)$. Without being too formal, we can create this localized function by taking the average of the magnetization in small boxes centered in $x$. 

With this we can define the Landau-Ginzburg free energy $F[\phi (x)]$ as a functional, i.e., a scalar function that has functions as an input. The partition function then becomes a path integral of the form

\begin{equation}
    \int \mathcal{D}\phi e^{-\beta F[\phi (x)]}.
\end{equation}

As a special example, for $B=0$ the system is symmetric under the transformation $s_i \rightarrow - s_i$. Then only quadratic terms of $\phi$ may appear in the free energy:

\begin{equation}
    F[\phi (x)] = \int d^dx \left[ \frac{1}{2}g_1\phi^2 + g_2\phi^4 + g_3\phi^6 + \frac{1}{2}\gamma (T) (\nabla \phi)^2\right] + ...
\end{equation}
This can be solved and we have that $g_1 \sim (T - T_c)$ and $g_2 \sim \frac{4}{3}T$, where $T_c$ is a special critical value of temperature.

It is helpful to write the Fourier transform of the magnetization

\begin{equation}
    \phi_k = \int d^dxe^{-ikx}\phi (x).
\end{equation}
The partition function then becomes

\begin{equation}
    \int \mathcal{D}\phi (x) = \prod_kd\phi_ke^{-F}.
\end{equation}

When we constructed $\phi (x)$, we used small boxes on the lattice of the space. Let's say the boxes are of size $a$, then there is a cutoff of the Fourier modes for $\Lambda \sim 1/a$

$$\phi_k = 0 \text{      for     } k>\Lambda.$$
Physics depend on the scale. We can set a scale of momentum by defining a second cutoff

\begin{equation}
    \Lambda^\prime = b \Lambda, \hspace{ 0.1\linewidth } b<1.
\end{equation}

We use this scale to separate the modes in 
\begin{enumerate}
    \item Long-wavelength fluctuations $\phi^-_k$ \\
    \begin{equation}
          \phi^-_k = \left\{
  \begin{array}{lr}
    \phi_k & : k < \Lambda^\prime\\
    0 & : k > \Lambda^\prime
  \end{array}
\right. .
    \end{equation}
    \item Short-wavelength fluctuations $\phi^+_k$ \\
    \begin{equation}
          \phi^-_k = \left\{
  \begin{array}{lr}
    \phi_k & : k < \Lambda^\prime\\
    0 & : k > \Lambda^\prime
  \end{array}
\right. .
    \end{equation}
\end{enumerate}

The modes $\phi^+$ are going to be integrated out since we do not care about their dynamics in this scale. The free energy decomposes into

\begin{equation}
    F[\phi_k] = F_0[\phi^-_k] + F_0[\phi^+_k] + F_I[\phi^-_k,\phi^+_k].
\end{equation}

The partition function is then

\begin{equation}
    Z =  \prod_{k<\Lambda^\prime }d\phi^-_ke^{-F_0[\phi^-_k]} \prod_{\Lambda^\prime < k < \Lambda } d\phi^+_ke^{-F_0[\phi^+_k]} e^{F_1[\phi^-_k,\phi^+_k]}.
\end{equation}

The second integral will generate new terms in an effective theory for the long-wave modes:

\begin{equation}
    Z =  \int \mathcal{D}\phi^- e^{-F^\prime [\phi^-]},
\end{equation}
where $F^\prime$ will have the same form as before (it is the most general form), but the constants will change:

\begin{equation}
    F^\prime [\phi ] = \int d^dxx \left[ \frac{1}{2}g_1^\prime \phi^2 + g_2^\prime\phi^4 + \frac{1}{2}\gamma^\prime (\nabla \phi )^2 \right].
\end{equation}

Two important steps are still to be made. First, the effective theory is being integrated up to $\Lambda^\prime$. We would like to rescale the space so that it can be integrated up to the original cutoff $\Lambda$, so that the theories can be compared with each other. For that we define the scale transformations

\begin{equation}
    x^\prime = xb \text{ and } k^\prime = k/b.
\end{equation}
For the same reason, we want the constant in front of the term containing the derivatives to be 1. For that we define a transformation of the modes

\begin{equation}
    \phi^\prime_k = \sqrt{\gamma^\prime }\phi^-_k.
\end{equation}

The dependence on the scale of the theory is now in the constants ($g_1$,$g_2$). The parametrize a space of theories that are equivalent, in the sense that they only differ on the scale they are being looked at. This family of theories has a particular name, the  \textit{renormalization group}. The exact form on how the couplings depend on the scale are equations called \textit{renormalization group flows}.

\section{RG Flows in QFT}

\subsection{Renormalization}

Quantum field theory aims to describe the dynamics of fundamental fields $\Phi_i$ of a theory through an action that is built up with local operators $\mathcal{O}_i$ of the fields. In general we will have an action like

\begin{equation}
    S[\phi_i] = \int d^D x \lambda_i \mathcal{O}_i,
\end{equation}
where $\lambda_i$ are numbers, called coupling constants, that measure the intensity of the operators. Different values of the constants define different theories. So we can parametrize a space of field theories with the set $\left\lbrace \lambda_i \right\rbrace$. For example, in a $O(n)$ models, to be studied in a later chapter, there are two local operators, so we can define a space of $O(n)$ theories that is 2-dimensional and parametrized by $\left\lbrace \lambda_1 , \lambda_2 \right\rbrace$. 

We can calculate observables with the use of the generating functional

\begin{equation}
    \mathcal{Z} = \int [\mathcal{D}\Phi ]e^{(-iS[\Phi])}.
\end{equation}
This is where singularities appear on certain actions. As we have described before, Wilsonian renormalization deals with the problem by defining an effective action for fields that are below a certain energy scale. First we define a cutoff $\Lambda$ in the generating functional

\begin{equation}
    \mathcal{Z} = \int_{\Lambda} [\mathcal{D}\Phi ]e^{-iS[\Phi]}.
\end{equation}
We then transform the fields to momentum space

\begin{equation}
    \Phi (x) = \int_\Lambda \frac{d^d p}{(2\pi)^d}e^{-ipx}\Phi (p).
\end{equation}
Separating the fields as slow modes $\Phi_S$, fields with $|p| < b\Lambda$, and fast modes $\Phi_f$, $b\Lambda |p| < \Lambda$, we get that the generating functional becomes

\begin{equation}
    \mathcal{Z} = \int_{\Lambda} [\mathcal{D}\Phi_s][\mathcal{D}\Phi_f]e^{-iS[\Phi_s+\Phi_f]}.
\end{equation}
Integrating out the fast modes gives us the effective theory

\begin{equation}
    \mathcal{Z} = \int_{\Lambda} [\mathcal{D}\Phi_s]e^{-iS^\prime [\Phi_s]}.
\end{equation}
Now we want to remove the dependence on the cutoff $\Lambda$. We can do that by rescaling the space and momentum variables

\begin{equation}
    x^\prime = xb \hspace{.1\linewidth} p^\prime = \frac{p}{b}.
\end{equation}
This inevitably leads to transformations on the fields

\begin{equation}
    \Phi = b^{-\Delta_\Phi}\Phi_s.
\end{equation}

This transformation is called scale transformation and $\Delta_\Phi$ is the scaling dimension of the field $\Phi$.  We define the transformation so that the kinetic

\begin{equation}\label{kinect}
    \mathcal{L_0} = \frac{1}{2}(\partial \phi)^2.
\end{equation}
term in the action is unchanged. So the dependence on the scale of the theory will be on the coupling constants $\lambda_i$. 

\begin{equation}
    S^\prime [\Phi] = \int d^D x \lambda^\prime_i \mathcal{O}_i.
\end{equation}

\subsection{Callan\textendash Symanzik equations}

The renormalized correlation functions of the theory

\begin{equation}
    G^{(n)}(x) = <\Omega | T \phi_1\phi_2 ... \phi_n | \Omega > = (-i\hbar )^n\left. \frac{1}{Z[0]}\frac{\partial^n Z}{\partial J(x_1)\partial J(x_2)...\partial J(x_n)}\right|_{J=0},
\end{equation}
where $J$ is an auxiliary function and $\Omega$ is the quantum state of the particles, depends on a scale of energy $\mu$ and on the coupling constants. Considering small shifts on these parameters
\begin{align}
    \begin{split}
        \phi & \rightarrow (1+\delta\eta)\phi,        \\
        \lambda & \rightarrow \lambda + \delta\lambda,        \\
        \mu & \rightarrow \mu + \delta\mu
    \end{split}
\end{align}
gets us the Callan\textendash Symanzik (CS) equations,

\begin{equation} \label{CSequ}
\left(\mu\frac{\partial}{\partial\mu}+\mu\frac{\partial\lambda_i}{\partial\mu}\frac{\partial}{\partial\lambda_i}-n\mu\frac{\partial\eta}{\partial\mu}\right)G^{(n)}=0.
\end{equation}

We may define the following quantities

\begin{equation} \label{defbeta}
\beta_i(\lambda) = \mu \frac{\partial\lambda_i}{\partial\mu},
\end{equation}
\begin{equation}
\gamma(\lambda) = - \mu \frac{\partial\eta}{\partial\mu}.
\end{equation}
We want to determine how the renormalized theory changes as we change the scale $\mu$, i.e. we want to know $\beta$ and $\gamma$. A general form for $\gamma$ can be found considering a field renormalization $Z$ \cite{Peskin}

\begin{equation}
\phi(\mu)=\sqrt{Z}(\mu)\phi_0.
\end{equation}
We shift the scale by an infinitesimal quantity $\delta\mu$

\begin{equation}
\phi+\delta\phi=\sqrt{Z}(\mu+\delta\mu)\phi_0.
\end{equation}

Dividing the last two equations, and recalling the definition of $\eta=\dfrac{\delta\phi}{\phi}$, we find a general expression for $\gamma$ in terms of $Z$:

\begin{equation}
\gamma =  -\mu \frac{\partial}{\partial\mu}\left(\frac{\delta\phi}{\phi}\right)   =  \frac{\mu}{2}\frac{1}{Z}\frac{\partial Z}{\partial \mu}.
\end{equation}

Now we can use the CS equation to find $\beta$ in terms of $Z$ and $G$,

\begin{equation}
\beta_i(\lambda)\frac{\partial G^{(n)}}{\partial \lambda_i} = \left(-\mu\frac{\partial}{\partial\mu}+n\gamma\right)G^{(n)}.
\end{equation}
Since $\beta$ and $\gamma$ are dimensionless quantities and the only scale of the system is $\mu$, they do not depend explicitly on it, there is no other dimensionfull parameter to cancel the dimensions of $\mu$ to form a dimensionless quan    tity. We may define $t=\ln \mu$ as a ``RG time", and (\ref{defbeta}) becomes

\begin{equation}\label{betadefi}
\beta_i(\lambda) = \frac{\partial\lambda_i}{\partial t}.
\end{equation}

This set of partial differential equations defines a dynamical system where a coupling $\lambda_i$ flows with velocity $\beta_i$ in a space $\Tau$ as the RG time changes. In this sense, $\Tau$ can be viewed as the space of field theories spanned by the operators $\left\lbrace\mathcal{O}_i\right\rbrace$. 

Lets discuss another thing regarding $\beta$-functions, before we start studying fixed points. As mentioned before, we define a transformation that lets the kinetic term \ref{kinect} unchanged. The kinetic action with the rescaled coordinates is

\begin{equation}
    S[\Phi] = \int d^dx\frac{1}{2}(\mu^{-d+2+2\Delta_\phi})(\partial \phi)^2.
\end{equation}
The transformation needed is $\phi \rightarrow \mu^{\frac{d-2-2\Delta_\phi}{2}}\phi$. A special case is when $\Delta_\phi = \frac{d-2}{2}$, which is a scale invariant theory. In general, the scaling dimension $\Delta_\mathcal{O}$ of an operator $\mathcal{O}$ is defined through the scale transformation giving

\begin{equation}
    \mathcal{O} \rightarrow \mu^{\Delta_\mathcal{O}}\mathcal{O}.
\end{equation}
As it is described in \cite{Naga}, we can redefine the coupling constants to become dimensionless, which gives a term to the $\beta$ function

\begin{equation}\label{yeahyeah}
    \bar{\beta_\mathcal{O}} = (\Delta_\mathcal{O}- d)\lambda_\mathcal{O} + \beta_\mathcal{O}.
\end{equation}

Operators such that $\Delta_\phi > d$ are called relevant since they are relevant perturbations under the process of renormalization. Similarly, operators such that $\Delta_\phi < d$ are called irrelevant and $\Delta_\phi = d$ are called marginal. These terms reflect the nature of the system as a dynamical one, and we will explore this more later.


\section{Conformal Invariance and Fixed Points}\label{confconfconf}

\subsection{Conformal Symmetry}

We start by assuming the symmetries of flat space, namely translations, rotations and boosts. These collectively are called the Poincar\'e transformations:

\begin{equation}
    x^{\prime\mu} = \Lambda^\mu_\nu x^\nu +a^\mu.
\end{equation}
A special set of points in the theory space are the ones for which $\beta (\lambda^{*})=0$. These points are known as fixed points and they represent field theories that are invariant under the renormalization procedure and are known as scale invariant theories, since they are invariant under scale transformations

\begin{equation}
    x^\prime=\mu x \text{    and     } p^\prime = p/\mu .
\end{equation}

In most cases, we can extend the symmetry group to include special conformal transformations (\cite{ScaleCobnfP}, \cite{ScaleCobnfY})

\begin{equation}
    x^{\prime i} = \frac{x^i - (x\dot x)a^i}{1-2(x\dot a) + (a\dot a)(x\dot x)}, \text{    for an arbitrary } a,
\end{equation}
and we call these theories as conformal field theories (CFT).

CFT's are also important in statistical mechanics\cite{Lectures}. They describe critical points on second order phase transitions. For example, let us remember the Ising model at a finite temperature $T$ introduced in Sec. \ref{IsingModelModel}. The magnetic field and the interaction form an order in the way spins are positioned. But as the temperature increases, the nice order is perturbed by the energy fluctuations. 

The correlation function between two magnetizations are given by

\begin{equation}
    \left< \phi (x)\phi (y) \right> \sim   \begin{dcases}
    \dfrac{1}{r^{d-2}} & : r << \xi\\
    \dfrac{e^{-r/\xi}}{r^{(d-1)/2}} & : r >> \xi
  \end{dcases}
 ,
\end{equation}
where $r=|x-y|$, and $\xi$ is a length scale known as the \textit{correlation length}. This means that distant spins will have much weaker correlations between spins (they fall off exponentially) than those that are closer (they fall off by a power law). Then, what we will see are ``patches", regions of the space that are highly correlated with the size of the correlation of $\sim \xi$. 

The correlation length depends on the temperature of the system,

\begin{equation}
    \xi = \dfrac{1}{|T-T_c|^{1/2}}.
\end{equation}
So for a particular critical value of the temperature the correlation gets infinite. This is a characteristic of second order phase transitions. At the exact critical value the correlation length is infinite. This means that the same fluctuations of energy occur at all scales, thus the system is scale invariant. 

A good example to see this effect is the Youtube video\footnote{Available at https://www.youtube.com/watch?v=MxRddFrEnPc} uploaded by Douglas Ashton \cite{youtube} where he simulated the Ising model with $2^{34}$ spins at three different temperatures, one smaller then $T_c$, one bigger then $T_c$ and one exactly at the critical temperature. As he changes the scales of the theory it is very clear that the one calculated at $T_c$ never loses its structure, just like a fractal doesn't loose its structure as it gets zoomed out or zoomed in.

\subsection{Conformal Algebra}

The conformal algebra is given by the following generators:

\begin{itemize}
    \item Translations $P_\mu=-i\del_\mu$,
    \item Lorentz generator $M_{\mu\nu}=i(x_\mu\del_\nu-x_\nu\del_\mu)$,
    \item Dilaton $D=-ix_\mu\del^\mu$,
    \item Special Conformal $K_\mu=i(x^2\del_\mu-2x_\mu x_\nu\del^nu)$.
\end{itemize}

With the following commutation relations,

\begin{itemize} 
    \item $[D,K_\mu]=-iK_\mu$,
    \item $[D,P_\mu]=iP_\mu$,
    \item $[K_\mu,P_\nu]=2i(\eta_{\mu\nu}D-M_{\mu\nu})$
    \item $[K_\mu,M_{\nu\rho}]=i(\eta_{\mu\nu}K_\rho-\eta_{\mu\rho}K_\nu)$,
    \item $[P_\mu,M_{\nu\rho}]=i(\eta_{\mu\nu}P_\rho-\eta_{\mu\rho}P_\nu)$,
    \item $[M_{\mu\nu},M_{\rho\sigma}]=i(\eta_{\nu\rho}M_{\mu\sigma}-\eta_{\mu\rho}M_{\nu\sigma}+\eta_{\mu\sigma}M_{\nu\rho}-\eta_{\nu\sigma}M_{\mu\rho})$.
\end{itemize}

These symmetries impose conditions on the stress energy tensor. The scale current is given by \cite{ScaleCobnfP}

\begin{equation}
    S^\mu (x) = x^\nu T_\nu^\mu (x) + K^\mu (x),
\end{equation}
where $T_\nu^\mu$ is the stress energy tensor and $K^\mu$ is an operator without explicit dependence on the coordinate. Conservation of the current means that

\begin{equation}
    T_\mu^\mu (x) = \del_\mu K^\mu (x).
\end{equation}

The conformal current is given by \cite{ScaleCobnfP},

\begin{equation}
    j^\mu_\nu (x) = v^\nu (x)T_\nu^\mu (x) + \del\cdot v (x) K^{\prime\mu} (x) + \del_\nu\del\cdot v(x)L^{\nu\mu} ,
\end{equation}
where $L^{\mu\nu}$ is an local operator and the vector field $v(x)$ satisfies

\begin{equation}
    \del_\mu v_\nu + \del\nu v_\mu = \frac{2}{d}g_{\mu\nu}\del\cdot v(x)
\end{equation}

For $d=2$ the equation above has harmonic solutions. This means that there is an extra condition on $L$ for that dimension and the conservation of the current gives,

\begin{equation}\label{2dimension}
    d = 2: \hspace{20pt} T_\mu^\mu  =  \del_\mu\del_\nu L^{\nu\mu},
\end{equation}
\begin{equation}\label{morethan2dimension}   
    d\geq 3: \hspace{20pt} T_\mu^\mu (x) = \del^2L(x).
\end{equation}
Then the following equivalent stress energy tensor can be defined as

\begin{equation}
    \Theta_{\mu\nu} = T_{\mu\nu} + \frac{1}{d-1}(\del_\mu\del_\nu L(x) - \eta_\mu\nu\del^2L(x)).
\end{equation}
This combined with condition (\ref{2dimension}), gives $\Theta^\mu_\nu=0$. Similarly, for $d\geq 3$ the stress tensor given by

\begin{align}
    \Theta^\prime_{\mu\nu} = & T_{\mu\nu} \nonumber \\ & + \frac{1}{d-2}(\del_\mu\del_\rho L^\rho_\nu(x) + \del_\nu\del_\rho L^\rho\mu(x) -\del^2L_{\mu\nu} - \eta_{\mu\nu}\del_\rho\del_\sigma L^{\rho\sigma} (x)) \nonumber\\
    & +\frac{1}{(d-1)(d-2)}(\eta_{\mu\nu}\del^2L_\rho^\rho (x) -\del_\mu\del\nu L_\rho^\rho (x)).
\end{align}
Together with conditions (\ref{morethan2dimension}), this gives $\Theta^{\prime\mu}_\mu=0$. We can then conclude that conformal symmetry means the existence of a traceless stress tensor.

CFT's are particularly interesting in 2 dimensions, where the number of conformal transformation is infinite \ref{2dimension}. This leads to an infinite algebra, called the \textit{Virasoro Algebra}, and it is sufficient constraints that makes this fixed points exactly solvable, without the need of perturbation theory\cite{Difrancesco}. The Virasoro Algebra is spanned by a countable infinite number of generators $L_m$ and the central charge $c$ with the following commutation relations.

\begin{equation}
    [L_m,L_n] = (m-n)L_{m+n}+\frac{c}{12}(m^3-m)\delta_{m+n.0}.
\end{equation}

\subsection{Correlation Functions}

The symmetries that are given in a conformal field theory lead to the fact that the correlation function between fields are very constrained. The definition of a correlations function of $n$-points is, given $\Phi=\left\lbrace\phi_1,...,\phi_n\right\rbrace$ fields,

\begin{equation}
    \left\langle\phi_1(x_1)...\phi_n(x_n)\right\rangle = \frac{1}{Z}\int [\mathcal{D}\Phi]\phi_1(x_1)...\phi_n(x_n)e^{-S[\Phi]},
\end{equation}
where $S$ is the action of the theory and it is invariant under conformal transformations. From here we will explore the possibilities for $n=2$.

The conformal algebra gives us the possibility to construct irreducible representations of the conformal group for a specific type of field. These are called \textit{quasi-primary fields} and they are defined, if they are scalar under Lorentz transformation, by the transformation rule under conformal transformations

\begin{equation}
    \Phi^\prime (x^\prime ) = \left| \dfrac{\del x^\prime}{\del x} \right|^{\frac{\Delta}{d}} \Phi (x),
\end{equation}
where the number $\Delta$ is the scaling dimension of the field. The 2-point correlation function for quasi primary field has to respect the relation

\begin{equation}
    \left\langle\phi^\prime_1(x^\prime_1)\phi^\prime_2(x^\prime_2)\right\rangle =\left| \dfrac{\del x^\prime}{\del x} \right|_{x^\prime_2}^{\frac{\Delta_1}{d} }\left| \dfrac{\del x^\prime}{\del x} \right|_{x^\prime_2}^{\frac{\Delta_2}{d}} \left\langle\phi_1(x_1)\phi_2(x_2)\right\rangle.
\end{equation}
Further constraints from the Poincar\'e group lead to

\begin{equation}
    \left\langle\phi_1(x_1)\phi_2(x_2)\right\rangle = f(|x_1-x_2|),
\end{equation}
and invariance under dilations implies that

\begin{equation}
    \frac{f(b|x_1-x_2|)}{f(|x_1-x_2|)} = b^{-(\Delta_1+\Delta_2)}.
\end{equation}
Therefore, for scale invariant theories we have

\begin{equation}
    \left\langle\phi_1(x_1)\phi_2(x_2)\right\rangle = \frac{C_{12}}{|x_1-x_2|^{\Delta_1+\Delta_2}}.
\end{equation}

The last symmetry to be considered comes from the special conformal transformations. This gives us the relation

\begin{equation}\label{scft}
    \left\langle\phi_1(x_1)\phi_2(x_2)\right\rangle = \frac{1}{\gamma_1^{\Delta_1}\gamma_2^{\Delta_2}}\frac{C_{12}}{|x_1-x_2|^{\Delta_1+\Delta_2}}(\gamma_1\gamma_2)^{\frac{\Delta_1+\Delta_2}{2}}.
\end{equation}
Where $\gamma = (1-2bx+b^2x^2)$. 

Gathering the results from scale invariance and the special conformal invariance, we see that $\Delta_1=\Delta_2$. Therefore

\begin{equation}
    \left\langle\phi_1(x_1)\phi_2(x_2)\right\rangle = \frac{C_{12}}{|x_1-x_2|^{2\Delta}}.
\end{equation}

We can do a similar construction for 3-point functions and the result we get is:

\begin{equation}
    \left\langle\phi_1(x_1)\phi_2(x_2)\phi_3(x_3)\right\rangle = \frac{C_{123}}{|x_1-x_2|^{\Delta_1+\Delta_2-\Delta_3}|x_1-x_3|^{\Delta_1+\Delta_3-\Delta_2}|x_2-x_3|^{\Delta_2+\Delta_3-\Delta_1}},
\end{equation}
where $C_{123}$ is a constant number.

\subsection{Operator Product Expansion}

An important characteristics of CFT's is that we can actually find a list of all conformal operators $\left\lbrace\mathcal{O}_i\right\rbrace$. Then, products of operators can be written as a linear combination of the basis $\left\lbrace\mathcal{O}_i\right\rbrace$.

Specifically in 2 dimensions we can define the complex variables of the space-time coordinates $z=x_1+iy_1$ and $w=x_2+iy_2$. Therefore, we can write the so-called operator product expansion (OPE):

\begin{equation}
    \mathcal{O}_i(z,\bar{z})\mathcal{O}_j(w,\bar{w}) = \sum_k C_{ij}^k(z-w,\bar{z}-\bar{w})\mathcal{O}_k(w,\bar{w}),
\end{equation}
where $C_{ij}^k$ are constants. In particular, for the stress energy tensor we have

\begin{equation}
    T(z)T(w) \sim \dfrac{c/2}{(z-w)^4}+\dfrac{2T(w)}{(z-w)^2}+\dfrac{\del T(w)}{(z-2)}.
\end{equation}

The number $c$ is called the \textit{central charge} of the CFT and it is a particularly important number. Generically speaking, the central charge measures the number of degrees of freedom in the theory. Also, it appears in the calculation of the \textit{trace anomaly} that arises when one introduces a macroscopic scale to the theory, breaking the conformal symmetry. The anomaly (in 2 dimensions) is given by

\begin{equation}
    T_\mu^\mu = c\dfrac{R(x)}{24\pi}.
\end{equation}

Where $R(x)$ is the Ricci scalar. For four dimensions the trace annomaly is calcuted to be

\begin{equation}
    T_\mu^\mu = \dfrac{1}{16\pi^2}\left(C_{\rho\sigma\kappa\lambda}C^{\rho\sigma\kappa\lambda}c+\tilde{R}_{\rho\sigma\kappa\lambda}\tilde{R}^{\rho\sigma\kappa\lambda}a\right),
\end{equation}
where $c$ and $a$ are constants, $C$ is the Weyl tensor and $\tilde{R}$ is the dual of the Riemann tensor. In odd dimensions the trace anomaly is absent.

\newpage


\chapter{Algebraic Topology}

One of the main mathematical tools to study manifolds and dynamical systems is called Morse homology. In this chapter we present a brief introduction to this subject. We start by stating that the objective is to find characteristics of a manifold that are going to uniquely defined it and family of manifolds that are related by homeomorphisms. These characteristics are called topological invariants and examples of such are conectedness, compactness and the Euler characteristic. A good reference for algebraic topology is the book \cite{Nakahara}. In the next sections we will construct a particular topological invariant called homology, which can be viewed as a generalization of the Euler characteristic. 

\section{The basics}

In order to get to more advanced topics like homology and homotopy, we are going to make sure that the basics are covered. We want to eventually discover properties that are intrinsic to spaces, but first let us properly define what we mean when we say intrinsic and space. We are going to start with the definition of a map

\begin{mydef}
    Given sets X and Y, a map $f$ from X to Y is a function that has elements of X as input and elements of Y as input. We denote this in the form
    $$f : X \rightarrow Y,$$
    $$f : x \rightarrow f(x),$$
    where $x \in X$. The set X is called the domain of f, Y is called the range of f and the set f(X) its called the image of f.
\end{mydef}

Maps can be thought of as transformations of a set into another. We are going to give special attention to maps that preserve some properties of the sets. For example, we have the following definitions

\begin{itemize}
    \item f is an injection if for every $x_1,x_2 \in X$ we have that $f(x_1)\neq f(x_2)$ if $x_1\neq x_2$,
    \item f is a surjection if for every $y\in Y$ there is some $x\in X$ such that $f(x)=y$,
    \item f is a bijection if it is a injection and a surjection.
\end{itemize}

Bijections have the property that the are invertible, i.e., for every bijection $f$ there is a map $f^{-1}$ such that

$$f^{-1} : Y \rightarrow X,$$
$$f^{-1} : f(x) \rightarrow x.$$

Now let us add some structure to the sets. The first definition we are going to make for that purpose is of a \textit{field}.

\begin{mydef}
    A \textit{field} F is a set together with two operations, one called addition, represented by the symbol ``$+$", and the other multiplication, represented by the symbol ``$\cdot $". These operations are maps such that, given a, b and c elements of the set,
    \begin{itemize}
        \item $a+b=b+a$ and $b\cdot a=a\cdot b$,
        \item $a + (b+c)=(a+b)+c$ and $a \cdot  (b\cdot c)=(a\cdot b)\cdot  c$,
        \item There exist elements 0 and 1 such that $a+0=a$ and $a\cdot 1=a$,
        \item For every a there exists a element $-a$ such that $a+ (-a)=0$ and for $a\neq 0$ there exists an element $a^{-1}$ such that $a\cdot a^{-1}=1$,
        \item $a\cdot (b+c)=a\cdot b + a\cdot c$.
    \end{itemize}
\end{mydef}

Now we can define another type of maps, called \textit{homomorphisms}, that preserve the structures of the sets. For example, if $f:X\rightarrow Y$ where $X$ and $Y$ are sets that have an addition structure, then $f$ is an homomorphim if $f(a+b)=f(a)+f(b)$, where $a+b$ is beign calculated using the addition from $X$ and $f(x)+f(y)$ is being calculated using the operation from $Y$. If $f$ is also a bijection. the we say that $f$ is an \textit{isomorphism} and that $X$ is \textit{isomorphic} to $Y$.

The word space, when used in this vague form, usually refers to a set with some structure defined in it. If we want to be more specific, we will need to go through an example of a type of spaces. The most prominent type of space is a vector space.

\begin{mydef}
A \textit{vector space} V over a field F is a set of elements (called vectors) in which the two operators of F are present. Suppose u, v and w are two elements of V and a and b are elements of F, then they must obey the following properties
\begin{enumerate}
    \item $v+u=u+v$,
    \item $(u+v)+w=u+(v+w)$,
    \item There exist a vector $0$ such that $v+0=v$,
    \item For every vector $u$ there exists a vector $-u$ such that $u + (-u)=0$,
    \item $a(v+u) = av+au$ and $(a+b)v = av+bv$,
    \item $(ab)v = a(bv)$,
    \item There exists an element of F $1$ such that $1u=u$.
\end{enumerate}
\end{mydef}

\section{Topology}

Vector spaces are the most common type of space that one will encounter when studying physics. But in order to get some more advanced mathematical results, we will define a more general space, called a \textit{topological space}

\begin{mydef}\textbf{Topology:}\\
A \textit{topology} $\mathscr{T}$ of a set X is a collection of subsets of X such that
\begin{enumerate}
    \item $\emptyset \in \mathscr{T}$ and $X\in\mathscr{T}$,
    \item $\mathscr{T}$ is closed under unions; if for every $i$ in an interval $I$ there is a set $U_i\in \mathscr{T}$, then $\bigcup_I U_i \in \mathscr{T}$,
    \item $\mathscr{T}$ is closed under finite intersections; if $U_i,U_j\in\mathscr{T}$ then $U_i\cap U_j \in \mathscr{T}$.
\end{enumerate}
\end{mydef}
Easy examples of topologies are if $\mathscr{T}$ contains all subsets of X or if $\mathscr{T}=(\emptyset,X)$. The former is known as the discrete topology and the latter is known as the trivial topology.

The pair $(X,\mathscr{T})$ is usually called a \textit{topological space}, although it is common to refer just $X$ as a topological space and $\mathscr{T}$ is called the topology of $X$.

An interesting example of topological spaces are the metric spaces. Given a set $X$, a metric $d: X\times X \rightarrow \mathcal{R}$ such that, given $x$, $y$ and $z \in X$,
\begin{itemize}
    \item Commutative: $d(x,y) = d(y,x)$,
    \item Positive Definite: $d(x,y) \geq 0$ and $d(x,y) = 0$ if and only if $x=y$,
    \item Triangle inequality: $d(x,y) + d(y,z) \geq d(x,z)$.
\end{itemize}

For any $x\in X$ we can define \textit{open balls} $B_\epsilon (x)$ of radius $\epsilon$ with,

\begin{equation}
    B_\epsilon (x) = \left\lbrace x^\prime \in X | d(x,y) \leq \epsilon \right\rbrace.
\end{equation}
Now, the union of all open balls forms a topology on $X$. This topology is known as the metric topology.

On topological spaces we can define interesting properties. For example,

\begin{mydef}
Given a topological space $(X,\mathscr{T)}$, we have that
\begin{itemize}
    \item Elements of $\mathscr{T}$ are called \textit{open subsets of X} ,
    \item A \textit{closed subset of X} is a set whose complement is a  open subset of X,
    \item For a given $x\in X$, an open neighborhood $U$ of x is a open set such that $x\in U$. A neighborhood of x $N$ is a subset of X (doesn't need itself to be open) such that $U \subseteq N$.
\end{itemize}
\end{mydef}

\begin{mydef}
A \textit{continuous map} $f:X\rightarrow Y$, where $X$ and $Y$ are topological spaces, such that for every open subset in Y, $U\subseteq Y$, the inverse image $f^{-1}(U)$ is open in X.
\end{mydef}

\begin{mydef}
A \textit{homeomorphisms} $f:X\rightarrow Y$ is a continuous bijection such that its inverse $f^{-1}:Y\rightarrow X$ is also continuous. If there is a homeomorphism between two topological spaces $X$ and $Y$, $X$ is said to be \textit{homeomorphic} to Y.
\end{mydef}

\section{Topological Invariants}

The definition of homeomorphisms gives us a direction in our goal to find structures that are fundamental to spaces. These mappings are so smooth that they can be seen as \textit{continuous deformations} of a space. For example, we can continuously deform a cube into a ball by ``curving" the edges, but we cannot continuously deform a ball into a ring, since for that we would need to make a hole. In fact, this intuition of continuous deformation matches the definition of homeomorphism. So we can consider spaces that are homeomorphic to each other as equivalent. In a more formal manner,

\begin{mydef}
An \textit{equivalence relation} $\sim$ is a relation with
\begin{itemize}
    \item Reflectiveness: $a\sim a$,
    \item Comutativity: If $a\sim b$ then $b\sim a$,
    \item Transitivity: If $a\sim b$ and $b\sim c$, then $a\sim c$.
\end{itemize}
\end{mydef}

Here, homeomorphism is a equivalence relation between topological spaces. Properties that are shared between equivalent spaces are called \textit{topological invariants}. Thus, if two spaces have different topological invariants, they are not homeomorphic. This is an alternative way to find out if two spaces are equivalent. Usually it might be hard to prove that no homeomorphism exists between them, but if they have different invariants we know they are not equivalent. Since we do not know the full set of topological invariants that there are\footnote{It may not even be possible to find complete list of topological invariants.}, we cannot make the opposite argument to prove that two spaces are equivalent.

Here we are going to give a small list of some topological invariants we know.

\begin{mydef}Hausdorff property\\
A topological space $X$ is said to be \textit{Hausdorff} if for any pair of $x$, $y \in X$ there exists two open sets $U_x, U_y$ such that $x\in U_x$ and $y\in Y$ and they are disjoint, i.e., $U_x\cap U_y = \emptyset$
\end{mydef}

\begin{mydef}Compactness\\
A topological space is said to be \textit{compact} if for every \textit{open covering} of X, i.e., a family $\left\lbrace U_i|i\in I\right\rbrace$ of open sets such that
$$\bigcup_{i\in I}U_i=X,$$
there is a finite set $J\subseteq I$ such that $\left\lbrace U_i|i\in J\right\rbrace$ is also a covering of X
\end{mydef}

\begin{mydef}Connectedness\\
A topological space $X$ is said to be \textit{connected} if there aren't two open disjoint subsets $X_1$ and $X_2$ such that $X=X_1\cup X_2$.
\end{mydef}

Here we have seen some examples of properties that help us understand what it means for two sets to be of the same homeomorphic class. In the next couple of sections we will see other more complex topological invariants.

\section{First Homotopy Group}

We hinted before that the number of holes in a space is a topological invariant. In order to make this more formal we first need to understand what we mean by holes. Our intuition behind this section it to define closed paths along the space, and we try to shrink the path to a point. If the path encloses a hole, then we will not be able to shrink the path, like in figure \ref{hmomotom}. Now we are going to make this more formal.

\begin{figure}[h]
\centering
\includegraphics[width=0.5\linewidth]{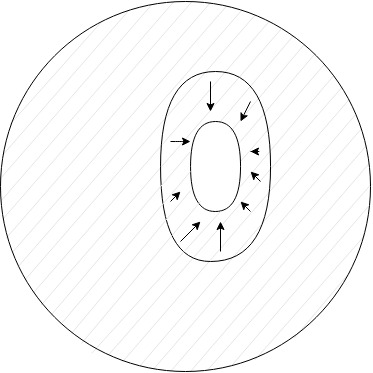}
\caption{A closed loop that encloses a hole can not be shrunken down to a point.}
\label{hmomotom}
\end{figure}

\begin{mydef}
Given a topological space X, a \textit{path} $\alpha$ is a continuous map along the space $$\alpha: [0,1]\rightarrow X.$$ The path has an initial point $x_0=\alpha (0)$ and an end point $x_1=\alpha (1)$. If they are the same, then the path is known as a \textit{loop}.
\end{mydef}

We can define a product of paths as a continuation of them. For $\alpha$ and $\beta$ paths in $X$ such that one end where the other starts, i.e., $\alpha (1)=\beta (0)$, then the product of the paths $\alpha *\beta$ is defined to be

\begin{equation}
    \alpha *\beta (t) = \left\{
  \begin{array}{lr}
    \alpha (2t) & : 0 \leq t \leq \frac{1}{2}\\
    \beta (2t-1) & : \frac{1}{2} \leq t \leq 1
  \end{array}
\right. .
\end{equation}
Also, it is possible to define the inverse of a path by $\alpha^{-1}(t)=\alpha (1-t)$.

We can now properly define what we mean by deforming a path into another.

\begin{mydef}
Given two loops $\alpha$ and $\beta$ on a topological space X that start/end on a point $x\in X$, a \textit{homotopy} between $\alpha$ and $\beta$
$$F:[0,1]\times [0,1] \rightarrow X$$
is a continuous map such that
\begin{align}
        F(s,0) = \alpha (s), \hspace{10pt} F(s,1)=\beta(s), \\
        F(0,t) = F(1,t) = x.
\end{align}
This means that the homotopy is a family of paths parameterized by s, and at $s=0$ the path is $\alpha$ and at the end $s=1$ the path is $\beta$. 

If a homotopy exists between two loops, they are said to be \textit{homotopic}.
\end{mydef}

Since homotopy defines an equivalence relation, we can say that two loops are equivalent if they are homotopic. As we have been hinted so far, spaces with no holes are such that all loops defined in them in the same point are equivalent, they all can be deformed to that point. Spaces with a hole will have loops that cannot be deformed to a point, this means that these loops are going to span a new equivalence class. In a formal and general way, 

\begin{mydef}
Given $X$ a topological space and $x\in X$, the \textit{first homotopy group} of X at x $\pi_1(X,x)$ is the set of all equivalent classes of loops at x. 
\end{mydef}

The first homotopy has the property of being topological invariant, though we are not going to prove that here.

The idea of homotopy can be extended to any map. Given $f,g$ two maps of the same domain and range, an homotopy is a continuous map such that 
\begin{equation}
    F(x,0)=f(x), \hspace{10 pt} F(x,1)=g(x).
\end{equation}
If there is a homotopy between maps $f$ and $g$, then we can say that they are homotopic.

\section{Homology of triangulable spaces}\label{dasdasxczczxvadfwewesdv86wrgi123jtr934tyhug}

When we study the geometry of polyhedra we usualy deconstruct the mathematical shape into its smaller componets, say the vertices, edges and faces. Simplexes are generalizations of those fundamental componets to any dimensional polytope. So a 0-simplex $\left( p_0\right)$ is a vertex, a 1-simplex $\left( p_0p_1\right)$ is an edge that conects the vertices $p_0$ and $p_1$, a 2-simplex $\left( p_0p_1p_2\right)$ is the face of a triangle, and generically we can write the $r$-simplex $\sigma_r$ as $\left( p_0p_1\ldots p_r\right)$. 

The fact that they are oriented means that we can define their inverses as any odd permutation, e.g., $\left( p_0p_1\right) = - \left( p_1p_0 \right)$. We can also define the $r$-boundary map of a $r$-simplex denoted by $\del_r$. As an example, we have $\del_1(p_0)=0$, $\del_2(p_0p_1)=p_1-p_0$, $\del_3(p_0p_1p_2)=(p_0p_1)+(p_1p_2)+(p_2p_0)$ and so on.

For $q\leq r$ we can define a $q$-face $\sigma_q$ of a $r$-simplex $\sigma_r$ as the $q$-simplex $\left( p_{i_0}p_{i_1}\ldots p_{i_q}\right)$ if the set of $\left\lbrace p_{i_j}\right\rbrace$ is in $\sigma_r$. If that's the case, we can write $\sigma_q \leq \sigma_r$.

A simplicial complex $K$ is a set of simplexes such that for every simplex in the set, all its faces are also in the set. The union of all simplexes in $K$, denote by $|K|$, defines a polytope in $\Re^n$. For a given topological space $X$, if $X$ is homeomorphic to a polytope $|K|$, then we say that $X$ is triangulable and  $|K|$ is a triangulation of X. 

For example, the empty triangle $\left\lbrace \left( p_0\right),\left( p_1\right),\left( p_2\right),\left( p_0p_1\right),\left( p_1p_2\right),\left( p_2p_1\right) \right\rbrace$ is a triangulation of the circle $S^1$ and the filled triangle $\left\lbrace \left( p_0\right),\left( p_1\right),\left( p_2\right),\left( p_0p_1\right),\left( p_1p_2\right) , \left( p_2p_1\right), \left( p_0p_1p_2\right) \right\rbrace$ is a triangulation of the 2-dimensional disk $D^2$, as can be seen in figure \ref{222222catra}. So from now on, we will analyze the topology of a space by looking to the topology of its triangulation.

\begin{figure}[h]
\centering
\includegraphics[width=0.5\linewidth]{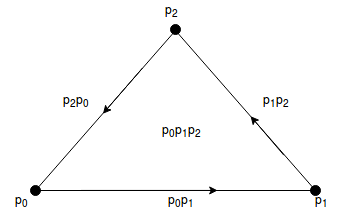}
\caption{The triangulation of a 2-dimensional disk is a filled triangle.}
\label{222222catra}
\end{figure}

We define the $r$-chain group $C_r(K)$ of a complex $K$ by a free abelian group generated by the $r$-simplexes of $K$. In general, if there are $I_r$ $r$-simplexes in $K$, then
\begin{equation}
C_r(K)=\underbrace{\mathbb{Z}\oplus\mathbb{Z}\oplus\cdots\oplus\mathbb{Z}}_{I_r}.
\end{equation}	

The set of integers accounts for any multiples of each simplex. Since the same space can have different triangulations, chain groups are not topological invariants. 

We can see that the boundary map is a map between chain groups,
\begin{equation}
C_r\xrightarrow[]{\del_r} C_{r-1}.
\end{equation}
Then, we can use this relations to define the set of topological invariants called homology groups
\begin{equation}
H_r(K)\coloneqq  \text{Kernel of }\del_r / \text{Image of }\del_{r+1},
\end{equation}
where the kernel of a map is the set of all the simplexes that map to 0
\begin{equation}
Z_r = ker(\del_r)=\left\lbrace c \in C_r(K) | \del_r(c)=0 \right\rbrace,
\end{equation}
and the image of the map is defined as
\begin{equation}
B_{r+1} = img(\del_{r+1})=\left\lbrace \del_{r+1}(c)  | c \in C_{r+1}(K) \right\rbrace,
\end{equation}

Essentially, the homology group $r$ measures the number of $r$-simplexes that do not have a boundary, but also are not the boundary of an $(r+1)$-simplex. 

\subsection{Calculation of the homology group}

Here we are going to give an example of how one would calculate the homology group of the 2-dimensional disk and sphere.

\begin{itemize}
    \item The 1d sphere\\
For the circle we have $K = \left\lbrace \left( p_0\right),\left( p_1\right),\left( p_2\right),\left( p_0p_1\right),\left( p_1p_2\right) , \left( p_2p_1\right)\right\rbrace$, the chain groups are 

\begin{align}
    C_0 = &\left\lbrace ip_0+jp_1+kp_2 | i,j,k \in \mathbb{Z} \right\rbrace \approx \mathbb{Z}\oplus\mathbb{Z}\oplus\mathbb{Z},\\
    C_1 = &\left\lbrace l(p_0p_1)+m(p_1p_2)+n(p_2p_0) | l,m,n \in \mathbb{Z} \right\rbrace \approx \mathbb{Z}\oplus\mathbb{Z}\oplus\mathbb{Z}.
\end{align}

Since there are no 2-simplexes, we have that the image of $\del_2$ is equal to zero, $B_2=0$. Since this is the case, $H_1=Z_1$. The kernel of $\del_1$ is not empty. In fact, the linear combination $l(p_0p_1)+l(p_1p_2)+l(p_2p_0)$ has no boundary. So $Z_1=\mathbb{Z}$ and $H_1=\mathbb{Z}$.

$H_0$ takes a special meaning. It essentially measures the number of connected pieces in the set. To see this proved see \cite{Nakahara} Since the sphere is a connected set, we have $H_0=\mathbb{Z}$.

\item The 2d disk\\
For the disk we have $K = \left\lbrace \left( p_0\right),\left( p_1\right),\left( p_2\right),\left( p_0p_1\right),\left( p_1p_2\right) , \left( p_2p_1\right), \left( p_0p_1p_2\right) \right\rbrace$, the chain groups are 

\begin{align}
    C_0 = &\left\lbrace ip_0+jp_1+kp_2 | i,j,k \in \mathbb{Z} \right\rbrace \approx \mathbb{Z}\oplus\mathbb{Z}\oplus\mathbb{Z}\\,
    C_1 = &\left\lbrace l(p_0p_1)+m(p_1p_2)+n(p_2p_0) | l,m,n \in \mathbb{Z} \right\rbrace \approx \mathbb{Z}\oplus\mathbb{Z}\oplus\mathbb{Z}\\,
    C_2 = &\left\lbrace a(p_0p_1p_2) | a \in \mathbb{Z} \right\rbrace \approx \mathbb{Z}.
\end{align}

Since $(p_0p_1p_2)$ is not a boundary of a bigger simplex, we have that $B_3=0$. So $H_2=Z_2$. But the boundary of $(p_0p_1p_2)$ is not zero, in fact $\del_2(p_0p_1p_2)=(p_0p_1)+(p_1p_2)+(p_2p_0)$. So $H_2=0$.

The story is similar for $H_1$, but there is a linear combination of 1-simplexes that have no boundary, $\del_1((p_0p_1)+(p_1p_2)+(p_2p_0))=0$, and so $Z_1=\mathbb{Z}$. But the image of $\del_2$ is not zero, in fact, as calculated above, the image contains the same linear combination $=(p_0p_1)+(p_1p_2)+(p_2p_0)$. Therefore $H_1=\mathbb{Z}/\mathbb{Z}=0$.

Since the disk is connected, we have that $H_0=\mathbb{Z}$.

\end{itemize}

One more thing to notice. We use the star notation to synthesize the information of the groups, $H_*=(H_0,H_1,H_2,...)$. Also, we use brackets to indicate the position of the non-zero objects in a sparse list, $\mathbb{Z}[n] = (0,...,\underbrace{\mathbb{Z}}_{\text{n-th entry}},0,...)$.

Some classical examples of homology groups are:\\

\begin{tabular}{c|cc}
\centering
    Disk $D^n$ & $(\mathbb{Z},0,\cdots)=$ &$\mathbb{Z}[0]$, \\
    Sphere $S^n$ & $(\mathbb{Z},0,\cdots,\mathbb{Z},0,\cdots)=$&$\mathbb{Z}[0]\oplus\mathbb{Z}[n]$,\\
    M{\"o}bius strip $S^n$ & $(\mathbb{Z},\mathbb{Z},0,\cdots,)=$&$\mathbb{Z}[0]\oplus\mathbb{Z}[1]$,\\
    Torus $T$ & $(\mathbb{Z},\mathbb{Z}\times\mathbb{Z},\mathbb{Z},0,\cdots)=$&$\mathbb{Z}[0]\oplus(\mathbb{Z}\times\mathbb{Z})[1]\oplus\mathbb{Z}[2].$
\end{tabular}

\subsection{Relative homology}

Another important definition that we are going to use is that of the relative homology. Let us say that the set A is a subset of X. It is clear from definition that $C_*(A) \subset C_*(X)$. Then we are led to define

\begin{itemize}
    \item The A-relative chain of X as $C_*(X)/C(A)$,
    \item The A-relative homology of X as $H_*(X,A)=H_*(X/A)$.
\end{itemize}

Relative homology is going to be essential when we will be defining Conley index in the next chapter.




\chapter{RG Flow and Dynamical Systems}

\section{Dynamical Systems}

\subsection{Historical Overview}

In mathematics, dynamical systems are systems that show how a quantity evolves with time. They were developed first by physicists as they were trying to solve Newton's mechanics. Newton's equation is a differential equation that determines how the position of a object changes in time, and its solution on systems with an inverse square external function was very successful at explaining the dynamics of gravitational bodies.

The Newton equation is given by

\begin{equation}
    m\ddot{x_i}=m\frac{d^2x_i}{dt^2} = F(x_1, x_2, ... ),
\end{equation}
where $x_i(t)$ indicates the position of the $i$-th object in space, $t$ is time and $F$ is an external force.\footnote{In this example we did not allow for $F$ to explicitly depend on time for simplicity reasons. In case we can not do that, we just consider $t$ to be an extra dimension of the space.} The fact that the system is of second order makes it different for the system seen in (\ref{betadefi}), but a simple trick can be used here. Let us define $v = \dot{x}$, then the system becomes

\begin{align}
    \dot{x_i} & =  v_i,   \nonumber \\
    \dot{v_i} & = \frac{F}{m}(x_1, x_2, ... ).
\end{align}
    
This system has twice the number of equations compared to the original one, but they are all first order coupled differential equations.

As physicists tried to take on more and more complex systems, they realised that it is not always that these differential equations are solvable. 

It is usually agreed that it was Henri Poincar\'e the creator of the mathematics of dynamical systems in the late XIX century. By shifting the focus from a quantitative description of the solutions to a qualitative one, he was able to answer questions like stability and asymptotic behaviours of systems that were not analytically solvable, like the famous three body problem. 

Throughout the XX century, the field was expanded thanks in part to the works of Aleksandr Lyapunov, George Birkhoff, Andrei Kolmogorov, Stephen Smale and many others who explored the relation of dynamical systems and Hamiltonian physics. Some important techniques were developed in this context, like bifurcation theory. In the 1960's, together with an advancement in computers, the American Edward Lorenz discovered chaotic behaviour, when the solutions of a system which depend on the initial conditions that have any uncertainty of measurements, lead to very different asymptotic behaviour. His work was a mark for the study of dynamical systems, popularizing its use for many varied applications from the weather forecast to biology, economy and the study of fractals.

To set the notation, we are going to use forward, given a set of functions $\left\lbrace \lambda_i (t)\right\rbrace$ a dynamical system is a system of differential equations of the form

\begin{equation}
    \dot{\lambda}_i = \beta_i (\lambda).
\end{equation}
A standard reference for dynamical systems is the book \cite{Strogatz}.

\subsection{Linearization of Flows}\label{Linear}

When the $\beta$ function is linear the problem becomes trivial. Of course, it is too much to ask that all systems of interest are linear but, since we are assuming smooth functions, if we zoomed in in a point in the space the equations becomes approximate to linear equations.

An interesting class of points in the theory space are the ones where $\beta = 0$, i.e., the flow is stationary. These points are called fixed points, and they are related with QFT's where the conformal symmetry holds (CFT's) as seen in Sec \ref{confconfconf}. Close to the fixed point, we can linearize the system doing the following procedure:

If $\lambda = (\lambda_1,\ldots,\lambda_n)$ then $\beta = (\dot{\lambda}_1,\ldots,\dot{\lambda}_n)$, and we set $\lambda^{*}$ as the fixed point, $\beta(\lambda^{*})=0$. Now we define $\tilde{\lambda}=\lambda-\lambda^{*}$, then

\begin{equation}
\begin{split}
\beta(\lambda) & = \beta(\lambda^{*}+\tilde{\lambda})\\
& = \beta(\lambda^{*})+J\tilde{\lambda} + O(\lambda^2)\\
& \approx J(\lambda-\lambda^{*}).
\end{split}
\end{equation}
where $J$ is the Jacobian matrix of $\beta$ calculated at the fixed point.

By looking at the linearized system we can find some characteristics of the fixed point. The eigenvalue of the Jacobian is related to the conformal dimensions of the operators in the system by $eigenvalue(J)_i=\Delta_i-d$, where $d$ is the space-time dimension like seen in equation (\ref{yeahyeah}).

Eigendirections of the Jacobian matrix with positive eigenvalue represent relevant operators, since small pertubations in the directions tend to grow. In dynamical systems we call these points \textit{unstable fixed points}. To represent this, we draw a diagram in the direction spanned by $\lambda$, with the unstable points represented by an empty circle with arrows pointing away from it

$$\leftarrow \Circle \rightarrow$$

Similarly, eigendirections with negative eigenvalue represent irrelevant operators. Flow lines close to this points will tend to move in the direction of the point, representing the fact that small perturbations in this direction are irrelevant in the renormalization process. This points are called \textit{unstable fixed points} and they are represented by filled circles with arrows pointing towards them

$$\rightarrow \CIRCLE \leftarrow$$

Eigendirections with eigenvalue equal to zero represent marginal operators. This points are called \textit{marginal fixed points}. Not all marginal fixed point are the same, they can be of various types. The most common marginal fixed point we are going to see in future sections are \textit{half stable} or \textit{semi-stable fixed points}, where coming from one side the point is seen as stable but from the other side of the same line the point is seen as unstable. They are represented on the space with a half filled circle and arrows according to their stability

$$\leftarrow \RIGHTcircle \leftarrow$$

When we are plotting the diagram for multiple operators, i.e, the space is not one-dimensional, we are not worried with the filling of the fixed points since the diagram would become too crowded. Instead we use arrows that represent the flow lines in order to represent the stabilities of a fixed point. If a fixed point has $\mu$ negative eigenvalues, the space of flows that tend towards the point, called \textit{stable manifold} or \textit{irrelevant manifold}, is $\mu$-dimensional. Similarly, we can define the \textit{unstable manifold} or \textit{relevant manifold} of a fixed point.

We could also have imaginary eigenvalues. These are expected in non-unitary theories since the conformal dimension takes complex values. The interpretation of such camplex eigenvalues is as follows:

\begin{enumerate}[(i)]
\item Real part positive: unstable spiral. Seen in figure \ref{Unspira}, the flows move away from a fixed point in spiral.
\item Real part negative: stable spiral. Seen in figure \ref{stspira}, the flows approach a fixed point in spiral.
\item Zero real part: center. Seen in figure \ref{cecece}, the flows form orbits around a fixed point.
\begin{figure}[h]
  \centering
  \begin{subfigure}[b]{0.45\linewidth}
    \includegraphics[width=\linewidth]{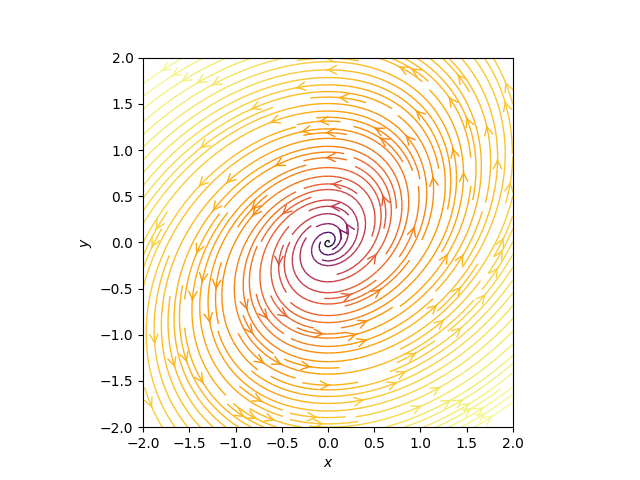}
     \caption{Unstable spiral, arrows pointing away from the fixed point in the center of the image.}
     \label{Unspira}
  \end{subfigure}
  \begin{subfigure}[b]{0.45\linewidth}
    \includegraphics[width=\linewidth]{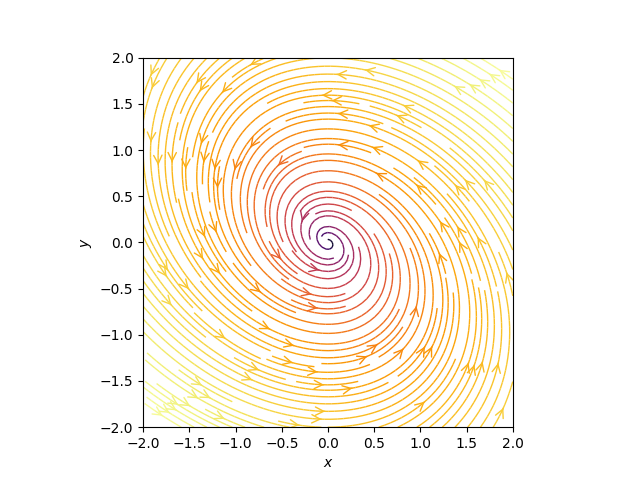}
    \caption{Stable spiral, arrows pointing toward the fixed point in the center of the image.}
    \label{stspira}
  \end{subfigure}
  \begin{subfigure}[b]{0.45\linewidth}
    \includegraphics[width=\linewidth]{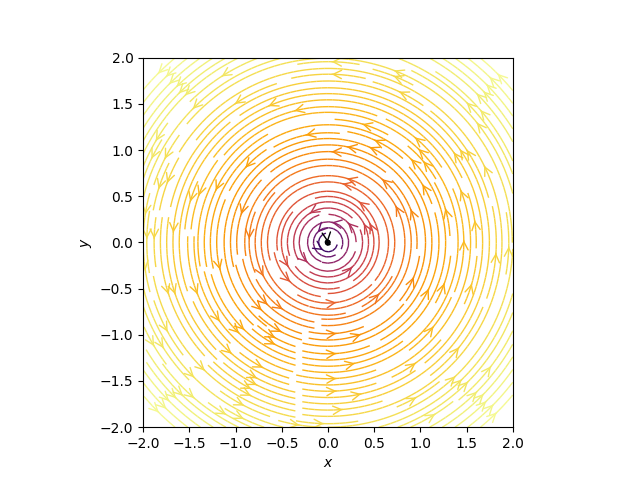}
    \caption{Center, arrows forming orbits around the finxed point in the center of the image.}
    \label{cecece}
  \end{subfigure}
  \caption{Possible behaviours near the fixed points for systems with imaginary eigenvalues of the Jacobian.}
\end{figure}
\end{enumerate}

\newpage

\subsection{Nonlinear asymptotic behaviours and limit cycles}

Not all possible behaviours are captured by the linearized system. Examples of these nonlinear behaviours are chaos and limit cycles.

Limit cycles are isolated closed trajectories, i.e., they represent periodic solutions of the system. The fact that they are isolated differentiates them from centers, since close to the cycle we can not have other closed trajectories. Instead, we have flows that tend to the cycle (that's why they are called limit) or tend to move away from it in a spiral.

\begin{figure}[h!]
\includegraphics[width=\linewidth]{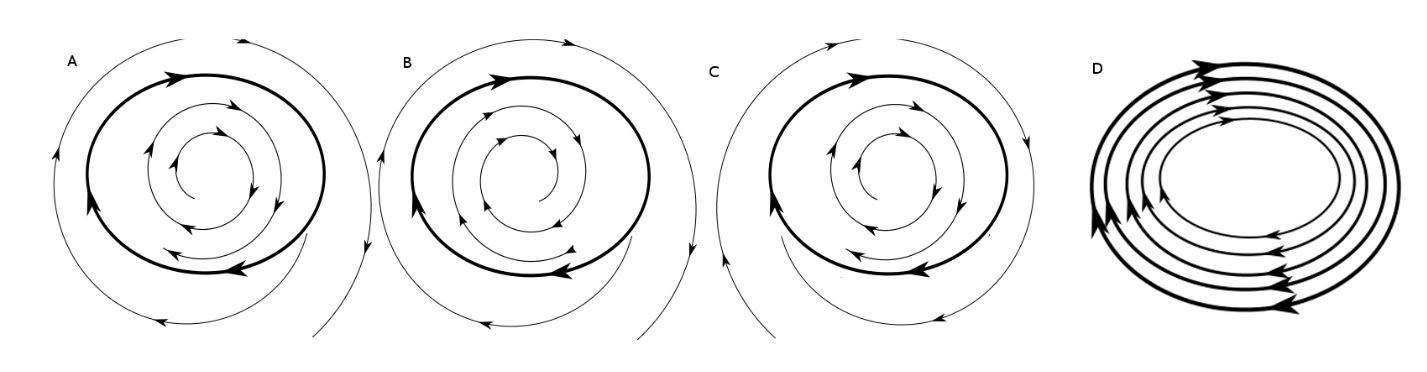}
\caption{(A): half-stable limit cycle, (B): unstable limit cycle, (C): stable limit cycle, (D): center. Picture taken from \cite{Martini}.}
\end{figure}

It is possible to rule out the existence of closed trajectories for some systems, though there is no universal method that works for every case. One objective of future sections is to show how one can decide whether the existence of periodic flows in certain systems is possible or not.

One of the most important theorems in dynamical systems has an unexpected consequence for limit cycles:
\newpage
\begin{theorem}\label{Intersecttheorem}
    Existence and Uniqueness: Suppose f is a function of class $C^1$ in some open connected set. Then the problem $\dot{x}=f(x),x(0)=x_0$ has a unique solution  $x(t)$ for some time interval $(-T,T)$
\end{theorem}

For a proof of the theorem see \cite{Uniquenessandthinhs}.

From the fact that the solution is unique we know that we can not have intersections of flows in the space. A corollary to this theorem is the fact that if a flow starts inside a region bounded by a closed orbit, we know that the flow will never leave this region.

\subsection{Gradient Systems}\label{Gradeint}

Suppose the flow can be written in the form 
\begin{equation}\label{Gradient}
\beta = -\nabla V,    
\end{equation}
for some continuously differentiable single-valued scalar function $V(\lambda)$. Such a flow is called gradient flow with potential function V.

\begin{theorem}\label{Teo1}
Closed orbits are impossible in gradient systems.
\end{theorem}

\begin{proof}
Suppose there were a closed orbit. There is a contradiction considering the change in V after one circuit. On one hand, $\Delta V=0$, since V is single valued. But on the other hand.
$$\Delta V = \int^T_0 \frac{dV}{dt}dt=\int^T_0(\nabla V \cdot \dot{\lambda})dt = - \int^T_0 \norm{\dot{\lambda}}^2dt < 0$$
unless $\dot{\lambda}=0$, which is a fixed point.
\end{proof}

The function V resembles the potential function of classical mechanics. In particular, it tends to decrease as the ``time'' progresses in the system. We can see this by looking at

\begin{equation}
    \dfrac{dV}{dt} = \dfrac{dV}{d \lambda}\dfrac{d\lambda}{dt}.
\end{equation}
By the definition (\ref{Gradient}) we have that 

\begin{equation}
    \dfrac{dV}{dt} = \dfrac{dV}{d\lambda}\left( -\dfrac{dV}{dt} \right) =  - \left( \dfrac{dV}{dt} \right)^2 \leq 0.
\end{equation}
In particular, all one dimensional dynamical systems are gradient.

\subsection{Lyapunov Function}\label{Lyplyp}

Consider a flow with a fixed point at $\lambda^{*}$. Suppose there is a Lyapunov function, i.e., a continuously differentiable real-valued function $V(\lambda)$ such that
\begin{enumerate}[(i)]
\item $V(\lambda)>0$ for $\lambda\neq\lambda^*$ and $V(\lambda^*)=0$,
\item $\dot{V}<0$ for all $\lambda\doteq\lambda^*$.
\end{enumerate}

\begin{theorem}
Any system such that a Lyapunov function can be defined has no closed orbits.
\end{theorem}

A proof can bee found in \cite{jordanobruno}.
Unfortunately there is no consistent method of finding Lyapunov functions or to prove they do not exist. One usually needs divine inspiration for that.

Another form to rule out the existence of limit cycles in 2-dimensional systems is the Bendixson-Dulac theorem.

\begin{theorem}\textbf{Bendixson-Dulac theorem}\\
    Let's suppose there is a scalar class $C^1$ function $\phi (\lambda)$ such that
    $$\nabla \cdot (\phi \beta),$$
    has only one sign $\neq 0$ in the entirety of a simply connected subset of the space. Then there are no periodic solutions for the system
    $$\dot{\lambda} = \beta (\lambda).$$
\end{theorem}

\begin{proof}
Assume the function $\phi (\lambda)$ exists. Assume now there is a limit cycle C in the region and denote D as the interior of C, then
$$\int\int_D \nabla \cdot (\phi \beta)d\lambda_1d\lambda_2 = \oint_C \phi (-\beta_2d\lambda_1+\beta_1d\lambda_2)=0.$$
where we used Green's theorem. This is a contradiction since $\nabla \cdot (\phi \beta)$ is strictly positive or strictly negative. 
\end{proof}

For curiosity, we show a similar theorem but with the opposite purpose. The Poincar\'e-Bendixson Theorem can be seen as a criteria for determining the existence of a limit cycle:

\begin{theorem}\textbf{Poincar\'e-Bendixson Theorem}\\
Suppose a system $\beta (\lambda) = \dot{\lambda}$ in a closed bounded set $R$ such that there is no fixed point in $R$. Then if there is a flow line $C$ that is confined in R, then C is a limit cycle or it asymptotically meets a limit cycle in R.
\end{theorem}

For a proof see \cite{Uniquenessandthinhs}.

\section{Index Theory}

We have seen in section \ref{Linear} how we can use the linearized version of the equations to gather information about fixed points. However, the linear system doesn't tell us anything of the \textit{global} behaviour of the flow. We present some method that were built to uncover this $global$ behaviour.

\subsection{Index of a Closed Curve}\label{inccloscur}

For a closed curve $C$ in a two-dimensional vector space with a vector field $\beta (\lambda)$, we can define an index that measures the winding of the vector field around $C$. For a more precise definition, let's say that for $\lambda \in C$, the vector field at this point $\beta (\lambda)$ makes an angle $\phi$ with one axis. Since $C$ in closed, the net difference of $\phi$ between a revolution in $C$, $[\phi]_C$, is a multiple of $2\pi$, then we can define the index as:

\begin{equation}
    I_C = \frac{1}{2\pi}[\phi]_C.
\end{equation}

Some interesting properties of the index of a curve are

\begin{enumerate}
    \item If $C^\prime$ is a continuous deformation of $C$, then $I_C=I_{C^\prime}$,
    \item If $C$ does not enclose a fixed point, then $I_C=0$,
    \item If $C$ is a trajectory of the system, i.e. a periodic flow, then $I_C=+1$.
\end{enumerate}

For every isolated fixed point $\lambda^*$, any closed curve $C$ that encloses $\lambda^*$ and no other fixed points, we have that $I_C$ is the same by property 1. Then we can drop the subscript $C$ and define $I$ as the index of the fixed point $\lambda^*$.

\begin{theorem}\label{coolctheorem}
If a closed curve C surrounds $n$ isolated fixed points, then 
$$I_C=I_1+I_2+...I_n,$$
where $I_k$ is the index of the k-th fixed point.
\end{theorem}

\begin{proof}
See chapter 6.8 of \cite{Strogatz}.
\end{proof}

For a given fixed point $\lambda^*$ we can take a $C$ such that the area enclosed by $C$ is a small neighborhood around $\lambda^*$. Then we can see that both stable and unstable fixed points have index $I=+1$ and for a saddle point\footnote{In two-dimensional phase space, a saddle point is a fixed point that is unstable in one direction and stable in another direction} the index is $I=-1$. Figure \ref{caramba} shows the three possibilities side-by-side.

\begin{figure}
\begin{center}
\includegraphics[width=0.9\linewidth]{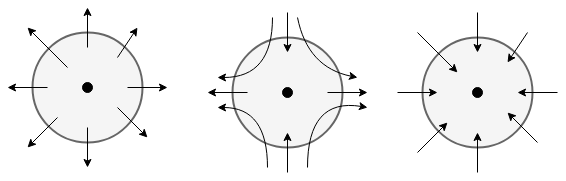}
\caption{From left to right, a stable, a saddle and an unstable fixed points. Note that the net change in the angle of the flow is $2\pi$ in a anti-clockwise orientation for the first and the last, while for the saddle the change is $2\pi$ in a clockwise orientation.}
\label{caramba}
\end{center}
\end{figure}

One application of this index is that if we can combine the theorem above with property 3 to find a way to rule out periodic orbits. If there is no way to construct a curve $C$ such that the indices of the fixed points enclosed by $C$ add to 0, then we can conclude that there are no periodic orbits in the space. A corollary is that for all periodic orbits, property 2 tell us that there must be at least one fixed point enclosed by the orbit.

\subsection{Conley Index}\label{conasldkla}

In the 70's, Charles Conley \cite{conleyOG} derived an important concept that applies algebraic topology to spaces defined by a vector field, the so-called Conley Index\cite{conley}. In order to define the Conley index, we need first to introduce the definition of an isolating neighborhood.

For a vector field $\beta$ and a theory space $\Tau$, such that $\beta : \mathbb{R}\times\Tau \rightarrow \Tau$, an isolating neighborhood is a compact set $N$ if:

\begin{equation}
    \text{Inv}(N,\beta) := \left\lbrace x \in N | \beta(\mathbb{R},x) \subset N \right\rbrace \subset \text{int} N.
\end{equation}

We call $S = \text{int} N $ an isolated invariant set. Now, for every set $S$, we can define a pair of sets $(N,L)$, called index pair, such that $L\subset N$ and:
\begin{itemize}
    \item $S=Inv(\overline{N\setminus L})$ and $N\setminus L$ is a neighborhood of S.
    \item Given $x\in L$ and $\beta([0,t],x) \subset N$, then $\beta([0,t],x)\subset L$.
    \item $L$ is an exit set in $N$; given $x\in N$ and $t_0 > 0$ such that $\beta(t_0,x)\notin N$, then there exists $0\leq t_1 < t_0$ such that $\beta(t_1,x\in L)$.
\end{itemize}

$L$ is said to be the exit set of $N$ because is the union of points in the boundary of $N$ such that the flow leaves $N$. We can define $L$ by the equation

\begin{equation}
    L = \left\lbrace x \in \partial N \mid \exists \text{ }t \in \mathbb{R}, \text{ }\beta(t,x)\cdot n(x) > 0 \right\rbrace,
\end{equation}
where $n$ is the unitary normal vector of the boundary $\partial N$.

We want to define an index, i.e., a quantity that is invariant under continuous transformations on $\text{inv}N$ or on $\beta$, that is

\begin{itemize}
    \item If $\text{inv}N = \text{inv}N'$, then $\text{Index}(N)=\text{Index}(N')$.
    \item if $\beta_r$ is a continuous family of vector fields, $r\in [0,1]$, then $\text{Index}(N,\beta_0)=\text{Index}(N,\beta_r)$.
\end{itemize}

A valid definition is when we collapse $L$ to a point forming a pointed set:

\begin{equation}
    h(S) = (N/L,[L]).
\end{equation}
See figure \ref{conley} for examples of constructions of pointed sets.
\begin{figure}[h]
    \centering
    \includegraphics[width=\linewidth]{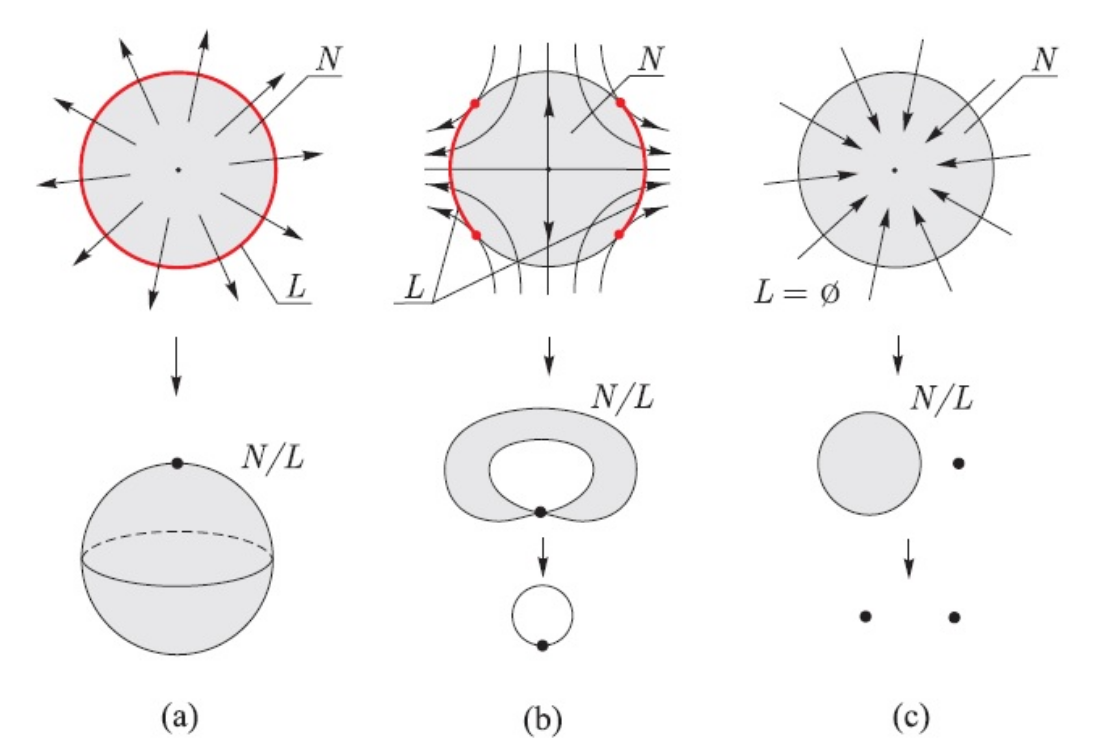}
    \caption{Examples of the construction of the pointed sets $(N/L,[L])$ for sets $N$ containing (a) an unstable fixed point, (b) a saddle fixed point and (c) a stable fixed point. Picture taken from \cite{oololo}}
    \label{conley}
\end{figure}

This definition is called the homotopy Conley index. We can also define a homological Conley index by taking the relative homology groups of $(N/L,[L])$:

\begin{equation}
    CH_*(S)=H_*(N/L,[L]).
\end{equation}

Now we make use of the fact that the Conley Index is a topological invariant. Suppose $N_1$ and $N_2$ are two isolating neighborhoods with the same isolating set for the vector field $\beta$,

\begin{equation}
    \text{Inv}(N_1,\beta)=\text{Inv}(N_2,\beta).
\end{equation}
Since we know that homologies are topological invariants, we get that the Conley Indices calculated for both sets must be equal:

\begin{equation}
    CH_*(N_1,L_1)=CH_*(N_2,L_2).
\end{equation}

Just like the index in the previous section was defined as the index of a curve but could be reinterpreted as the index of the fixed point enclosed by that curve, the Conley index can be seen as the index of the isolated invariant set $S$. If $S$ is a point, said a fixed point, such that the number of unstable eigendirections is $\mu$, then the Conley index is:

\begin{equation}
    CH_*=\mathbb{Z}[\mu].
\end{equation}

The value $\mu$ can be identified as the Morse-index of the fixed point as defined in (\ref{MorseIndex}). This definition is very useful since we have the following interesting property. If $S_1$ and $S_2$ are disjoint sets, and $S=S_1\cup S_2$, then

\begin{equation}
    CH_*(S) = CH_*(S_1) \oplus CH_*(S_2).
\end{equation}

Another way we can use the invariance of the index is if there is a continuous deformation of the vector field. Let us say $\beta$ is a smooth function of a \textit{control paremeter} $r$. Then we have that

\begin{equation}
    CH_*(\text{Inv}(N,\beta(0)) = CH_*(\text{Inv}(N,\beta(r)).
\end{equation}

Continuous deformation of the $\beta$ function will play a fundamental role in bifurcation theory, which will be studied later.


\subsection{The Conley Index of a 2-dimensional quadratic system}

In this section we will see all possibilities for the Conley index of an RG flow with two coupling constants (two RG equations) that have at most quadratic terms. This simple example will be useful to understand how we can use the index to infer about the structure of the system inside a compact set in theory space. The general structure of the system is:

\begin{align}
    \begin{split}
        \beta_1 = & a_{11}\lambda_1 + a_{12}\lambda_2 + b_{111}\lambda_1^2 + b_{122}\lambda_2^2 + b_{112}\lambda_1\lambda_2, \\
        \beta_2 = & a_{21}\lambda_1 + a_{22}\lambda_2 + b_{211}\lambda_1^2 + b_{222}\lambda_2^2 + b_{212}\lambda_1\lambda_2 .
    \end{split}
\end{align}

For simplicity, we choose $N$ to be the disc of radius $r$ and we parametrize the couplings to polar coordinates $(\lambda_1,\lambda_2)=(r\cos \phi,r \sin \phi)$. Then the equation $\beta\cdot n = 0$, where $n$ is the normal vector of the line tangent to $\del N$, has at most 6 solutions. $L$ can then be a set of a maximum of three disconnected intervals $I$ of $\partial N$, and we have 5 possibilities of different Conley indices:

\begin{enumerate}
    \item \textbf{$CH_* = \mathbb{Z}[0]$}, occurs when $L=\emptyset$. \\
    One example is if $N$ has a stable fixed point. 
    \item \textbf{$CH_* = \mathbb{Z}[2]$}, occurs when $L=\partial N$. 
    \\One example is if $N$ has an unstable fixed point.
    \item \textbf{$CH_* = 0$}, occurs when $L=I$. 
    \\One example is if $N$ has no fixed points .
    \item \textbf{$CH_* = \mathbb{Z}[1]$}, occurs when $L=I\cup I$. \\
    One example is if $N$ has a saddle fixed point.
    \item \textbf{$CH_* = \mathbb{Z}[1]\oplus\mathbb{Z}[1]$}, occurs when $L=I\cup I\cup I$. \\
    Examples of this must involve more than on fixed point in $N$.
\end{enumerate}

This list can be seen in \cite{GukovBif}. Notice that the cases in the list are only examples that gives the listed index, different configurations of fixed points inside of $N$ can give the same Conley-Index. So the usefulness of the index is to rule out what can not occur inside of $N$. For example, if the Conley index is different than zero, we know that there must be at least one fixed point. In fact, this is a general rule for the Conley index, not just for the system studied in this section.


\section{Bifurcation Theory}

In studying dynamical systems, one can ask the question of how a system changes with a specific control parameter $r$, i.e., a parameter that we control in the theory. Examples of control parameters in field theories are number of colors, number of flavors, space-time dimension or even the magnetic field in the Ising model. As this parameter is changed, one can have that the system undergoes a bifurcation, a form of a phase transition, at a critical value $r_0$ that fundamentally changes the behaviour of the flux on the phase space. We can classify bifurcation between 'local bifurcations', the ones that involve fixed points and can be studied by linearization methods, and global bifurcations, the ones that happen outside the range of the linear phenomena.

In this section, we will look at some examples of bifurcations that occur in RG-flows, as well as other examples that might be useful for a complete understanding of this topic.

\subsection{Saddle-node bifurcations}

In \cite{ConfLoss}, the author argues that a way for a field theory to 'lose its conformality', i.e., to lose its fixed points, is by what is called a saddle-node bifurcation, which is when two fixed point move along the space and collide, annihilating each other and leaving a 'bottle-neck', a slow moving flux, in their place. We can see this example of a local bifurcation by looking at this classical toy model:

\begin{equation}\label{saddle}
\beta (\lambda ) = r - \lambda^2.
\end{equation}

\begin{figure}[h]
\begin{center}
\includegraphics[width=0.8\linewidth]{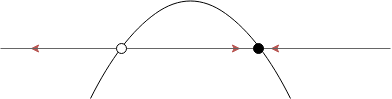}
\caption{ The phase space of $\beta$ for a positive $r$. The arrows indicate the directions of the flow.}
\label{oisitiv}
\end{center}
\end{figure}

The fixed points are $\lambda_\pm = \pm \sqrt{r}$ for $r > 0$ as can be seen in figure \ref{oisitiv}. The fixed points merge into a single one at the critical value $r=0$ and for $r<0$ the fixed points vanish into the complex plane. We can look at the linearized equation

\begin{equation}
\frac{\partial\beta}{\partial \lambda}=  - 2 \lambda
\end{equation}
to understand that the fixed point at $\lambda_+$ ($\lambda_-$) is a stable (unstable) fixed point and that these fixed points become marginal at the critical point before their disappearance.

A useful tool for studying bifurcations is the bifurcation diagram. We plot the position of the fixed point in respect to the control parameter $\lambda (r)$ and we use solid lines to denote stable fixed points and dashed lines to denote unstable fixed points, see figure \ref{putasdas} for the bifurcation diagram of the saddle-node.

\begin{figure}[h]
\includegraphics[width=\linewidth]{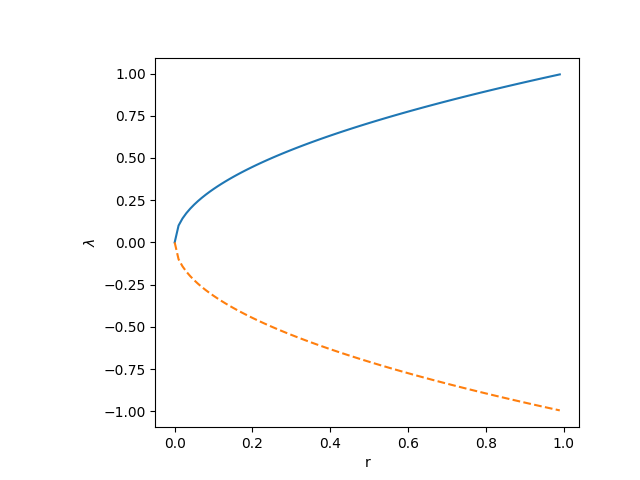}
\caption{Bifurcation diagram of a Saddle-Node Bifurcation. The solid line represents a stable fixed point and the dashed line represents an unstable fixed point.}
\label{putasdas}
\end{figure}

\subsection{Transcritical bifurcation}

Another classical example of a local bifurcation is the \textit{transcritical bifurcation}, see figure \ref{trabstrabsad}. It is characterized by the exchange of the stability properties of two fixed points. An example where this bifurcation can occur is

\begin{equation}
\beta (\lambda ) = - r\lambda + \lambda^2.
\end{equation}

The two fixed points in this system is $\lambda_0 = 0$ and $\lambda_r=r$. The linearized system, given by

\begin{equation}
\frac{\partial\beta}{\partial \lambda} = - r + 2\lambda,
\end{equation}
shows that for $r<0$, $\lambda_0$ is unstable and $\lambda_r$ is stable. When r becomes positive, the stability of both fixed points switches. At the critical value $r_{crit}=0$ the fixed points coincide and they cross through marginality.

\begin{figure}
\begin{center}
\includegraphics[width=0.8\linewidth]{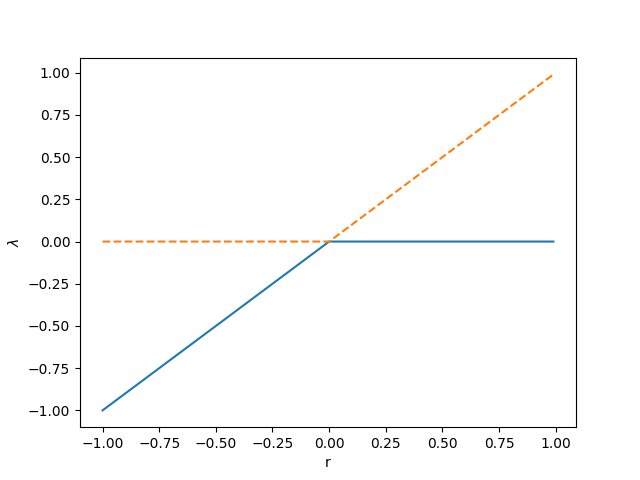}
\caption{Bifurcation diagram of a Transcritical Bifrucation.}
\label{trabstrabsad}
\end{center}
\end{figure}

\subsection{Pitchfork Bifurcation and Unfolding}

Consider the following system

\begin{equation}\label{Trans}
    \beta (\lambda) = r\lambda-\lambda^3.
\end{equation}
Independently of the value of $r$, there is always a trivial fixed point $\lambda_0=0$. For $r>0$, there are other two fixed points at $\lambda_\pm=\pm\sqrt{r}$. Also, by linearization, we see that the trivial point changes its stability properties at the bifurcation point $r=0$.

\begin{figure}
\begin{center}
\includegraphics[width=0.8\linewidth]{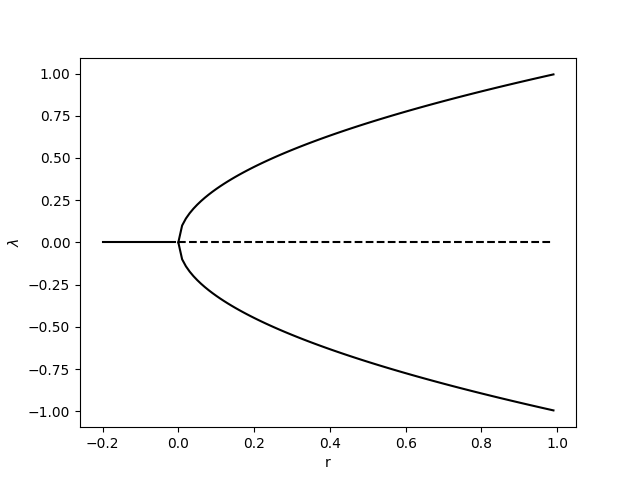}
\caption{Bifurcation diagram of a Supercritical Pichfork Bifrucation.}
\label{supertruper}
\end{center}
\end{figure}

The bifurcation above is known as the \textit{supercritical pitchfork bifurcation}, since there is a three-way bifurcation at the critical value, see figure \ref{supertruper}. Another version of the pitchfork is the subcritical, given by the system 

\begin{equation}
    \beta(\lambda)=-r\lambda+\lambda^3,
\end{equation}
where the stability of the fixed point are inverted from the previous example, see figure \ref{subtrooper}.

\begin{figure}
\begin{center}
\includegraphics[width=0.8\linewidth]{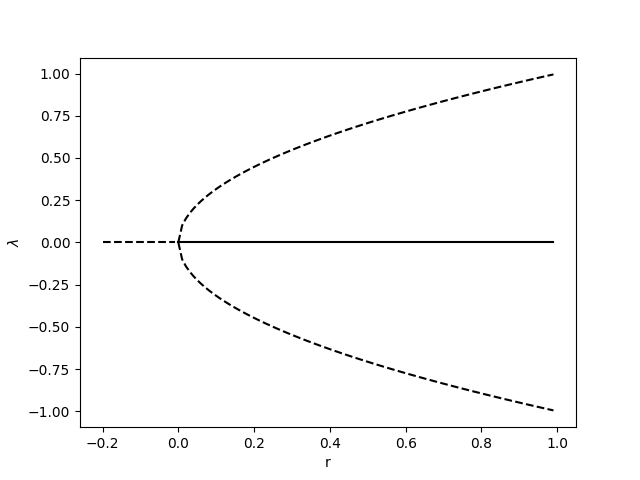}
\caption{Bifurcation diagram of a \textit{Subcritical Pichfork Bifrucation}.}
\label{subtrooper}
\end{center}
\end{figure}

The pitchfork bifurcation is a good example to introduce notions of stability of the bifurcation diagrams. Consider we add a parameter, say $h$, in the equation of a supercritical,

\begin{equation}
    \beta (\lambda) = h + r\lambda-\lambda^3 .
\end{equation}

For $h\neq 0$, what we have is that the pitchfork is separated, the once trivial fixed point does not meet the other two fixed points at a critical value, see figure \ref{imperfectrooper}. For this reason, this system is said to have an 'imperfect bifurcation' and the number $h$ is called an \textit{imperfection parameter}. Depending on the value of $h$, we can have three different scenarios depending of its relation with $h_c=\frac{2r}{2}\sqrt{\frac{r}{3}}$:

\begin{itemize}
    \item $|h|>h_c(r)$, there is only one fixed point in the system,
    \item $|h|=h_c(r)$, there are two fixed points in the system,
    \item $|h|<h_c(r)$  there are three fixed points in the system.
\end{itemize}

\begin{figure}
\begin{center}
\includegraphics[width=0.8\linewidth]{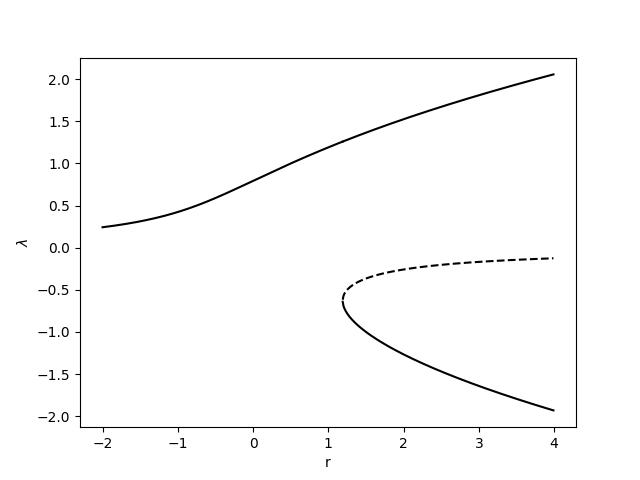}
\caption{Bifurcation diagram of an \textit{Imperfect Supercritical Pitchfork Bifurcation}. Note how the trivial fixed point 'misses' the bifurcation, and thus the bifurcation becomes a saddle-node.}
\label{imperfectrooper}
\end{center}
\end{figure}

The pitchfork is said to have 'codimension 2', codimension denoting the number of parameters needed in order to achieve the bifurcation. In this case, the pitchfork is achieved when $(r,h)=(0,0)$. Of the bifurcations we gave seen so far, the saddle is of codimension one and the transcritical is of codimension two. As seen before, however, both the pitchfork and the transcritical can be written using only one parameter. When this happen, higher codimension bifurcations described with few parameters, are called degenerate, since any small perturbation 'unfolds' the diagram and the tipe of bifurcation is changed. In perturbative RG equations, it is important to study the stability of the bifurcation diagrams for us to see if a higher loop correction in the $\beta$-equations destroys or maintains said bifurcations.

We can also see unfolding happening in the transcritical bifurcation (\ref{Trans}). If we add an imperfection parameter $h$ to the equation, we have two ways the transcritical bifurcation can unfold:

\begin{enumerate}
    \item $h>0$ : Both fixed points never meet, no bifurcation happens,
    \item $h<0$ : Two saddle fixed points occur at different values of $r$.
\end{enumerate}

\begin{figure}[h]
    \centering
    \begin{subfigure}[b]{0.45\textwidth}
        \includegraphics[width=\textwidth]{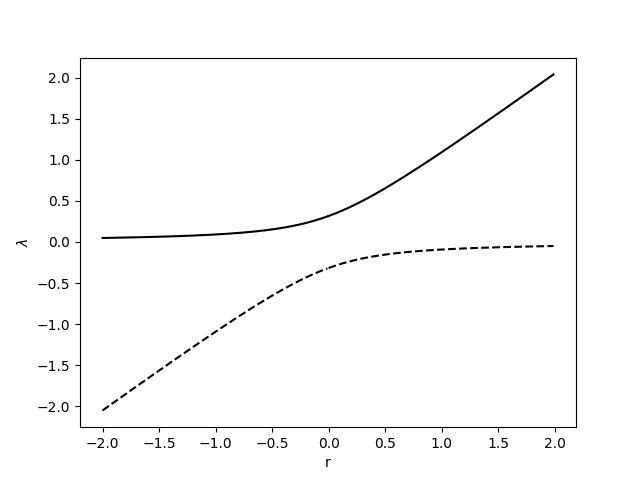}
    \end{subfigure}
    ~ 
    \begin{subfigure}[b]{0.45\textwidth}
        \includegraphics[width=\textwidth]{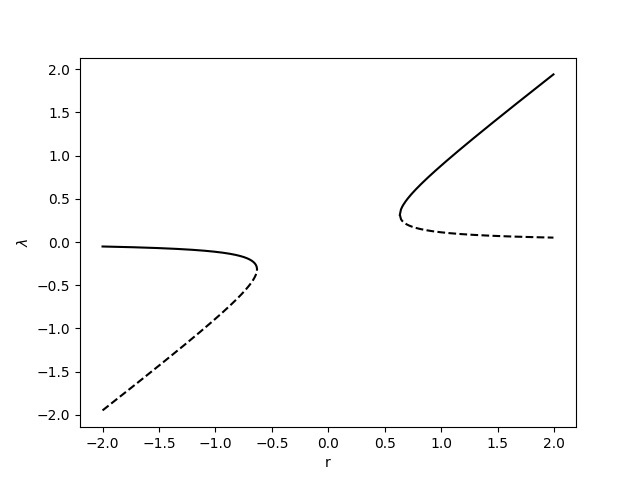}
    \end{subfigure}
    \caption{Unfoldings of the Transcritical bifurcation}\label{tranun}
\end{figure}

\subsection{Bifurcations and limit cycles}

Some bifurcations may lead to of periodic flows in theory space. This phenomenon may help us understand how theories may lose its gradient properties as a continuous parameter is changed. An interesting example, which is going to be studied later, is the one of Efimov Physics, which is believed to have a saddle-node bifurcation that leads to a limit cycle \cite{ConfLossEffi}.

\subsection{Hopf Bifurcation}

The \textit{Hopf bifurcation} is an example of a local bifurcation of codimension 1 that involves a limit cycle. It is easier to see the bifurcation by using radial coordinates in two dimensions, $r^2=\lambda_1^2+\lambda_2^2$ and $\theta = \tan (\frac{\lambda_2}{\lambda_1})$. To avoid confusion we will use the letter $\mu$ for the control parameter. We can see the Hopf bifurcation in the following system

\begin{align}
\label{eqn:eqlabel}
\begin{split}
    \beta_r = & - \mu r + r^3, \\
    \beta_\theta = & 1.
\end{split}
\end{align}

There is a trivial fixed point in $r=0$ for all values of $\mu$, and since $\dot{\theta}$ is never zero, the trivial fixed point is the only one in the system. For positive $\mu$, there is a values of $r$ for which $\beta_r$ is zero, namely $r=\sqrt{\mu}$. This $r$ corresponds to a limit cycles, and as the value of $\mu$ decreases the limit cycle meets the trivial fixed point and is removed (or, if we go from negative to positive $\mu$, we say that there is a birth of a limit cycle at $\mu=0$). Figure \ref{Hopftrooper} show this bifurcation happening.

\begin{figure}[h]
    \centering
    \begin{subfigure}[b]{0.47\textwidth}
        \includegraphics[width=\textwidth]{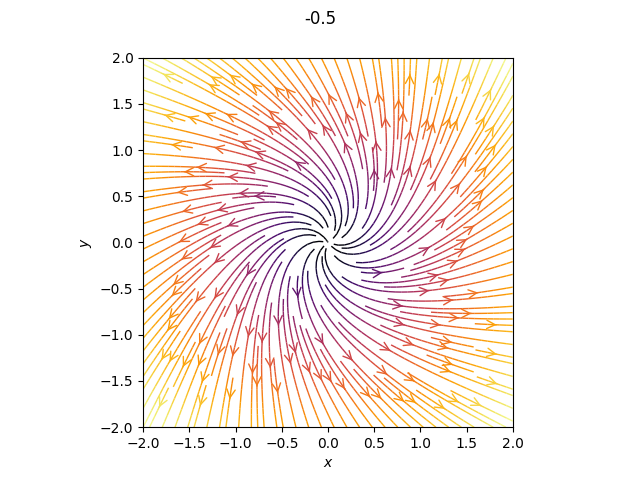}
        \caption{$\mu<0$}
    \end{subfigure}
    ~ 
    \begin{subfigure}[b]{0.47\textwidth}
        \includegraphics[width=\textwidth]{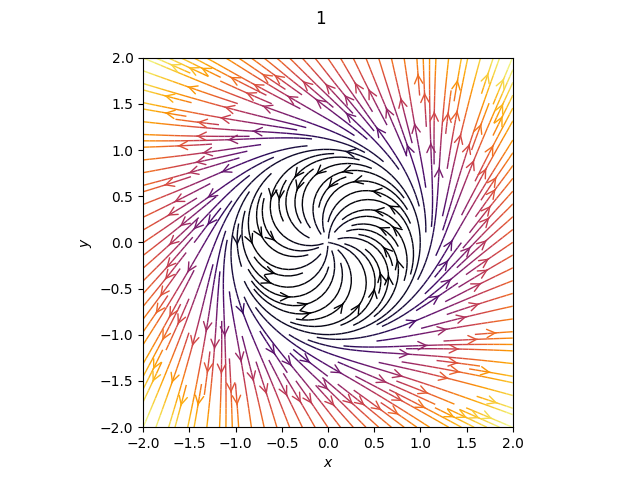}
        \caption{$\mu>0$}
    \end{subfigure}
    \caption{Example of a subcritical Hopf Bifurcation}
    \label{Hopftrooper}
\end{figure}

\subsection{Global bifurcations and Limit Cycles}\label{globalglovalgoaskock}

There are other ways that limit cycles can occur in bifurcations. However, these other ways involve global phenomena and thus are harder to find \cite{Strogatz}. A list of these bifurcations is given below:

\begin{itemize}
    \item Saddle-Node of limit cycles: It occurs when two limit cycles meet and annihilate. It can be seen as a modification of the Hopf Bifurcation, where the negative branch of the equation to find the limit cycles $\beta_r=0$ is not neglected.
    \item Homoclinic Bifurcation: It occurs when a limit cycles collides with a saddle point. At the critical value of the bifurcation, the flow becomes what is know as a 'homoclinic cycle', which is an infinite period flow.
    \item Saddle-node Infinite Period (SNIPER): A mix between global and local bifurcation. It happens when two fixed points meet and annihilate, leaving a limit cycle in its place.
\end{itemize}

A system that shows the Saddle-Node of cycles is:
\begin{align}\label{cacete}
\begin{split}
    \beta_r = & - \mu r + r^3 - r^5,\\
    \beta_\theta = & 1.
\end{split}
\end{align}

For this system, $\beta_r$ has two zeros, other than the trivial one, for $\mu$ in the range $0<\mu<1/4$, and at the exact critical value $\mu_c=1/4$ the two zeros merge and annihilate. This implies that $\beta_r$ has a saddle-node bifurcation. In the two-dimensional space we see that the saddle-node is between limit cycles. See figure \ref{dasdasdsadsadasdqweqrqrzx}.

\begin{figure}
\centering
\includegraphics[width=0.75\linewidth]{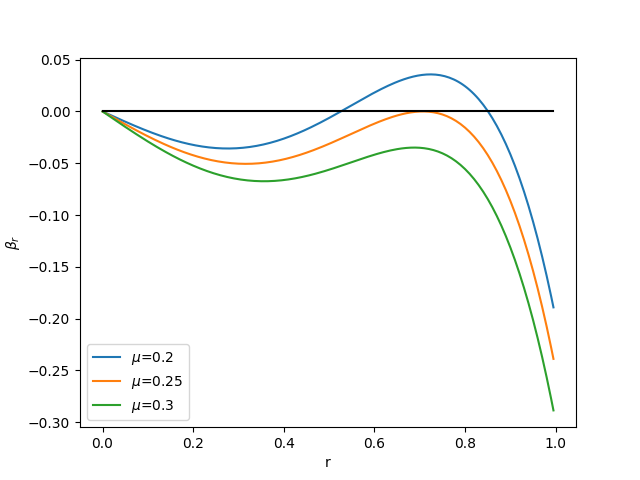}
\caption{Graphs of $\beta_r(r)$ of the system in (\ref{cacete}) for different values of $\mu$, below, above and at exactly the critical value $\mu_c=1/4$. This shows a saddle-node of the limit cycle when $\mu=0.25$.}
\label{dasdasdsadsadasdqweqrqrzx}
\end{figure}

One example of SNIPER bifurcation is the system:
\begin{align}
\begin{split}
    \beta_r = & r - r^3, \\
    \beta_\theta = & \mu + \cos(\theta ).
\end{split}
\end{align}

In this example, there are saddle-node bifurcations at $\mu=\pm 1$, and in the region $|\mu| > 1 $ there is a limit cycle in their place. See figure \ref{ASDKJOASUBDVASTGIDHOJIAPSKDASNHDGYTASGYHJKM}.

\begin{figure}
    \centering
    \begin{subfigure}[b]{0.6\textwidth}
        \includegraphics[width=\textwidth]{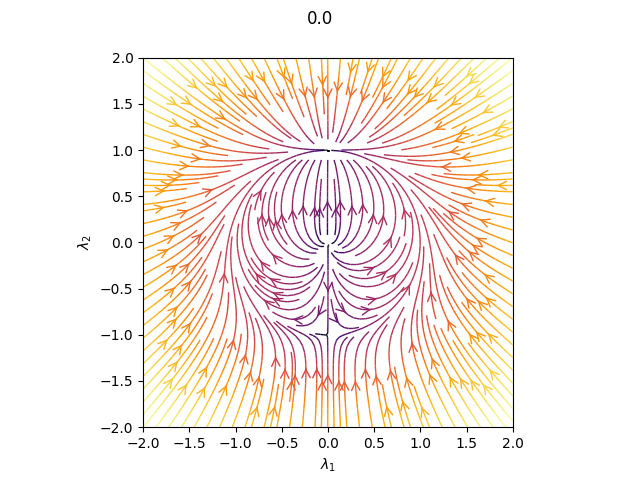}
        \caption{$\mu=0$}
    \end{subfigure}
    ~ 
    \begin{subfigure}[b]{0.6\textwidth}
        \includegraphics[width=\textwidth]{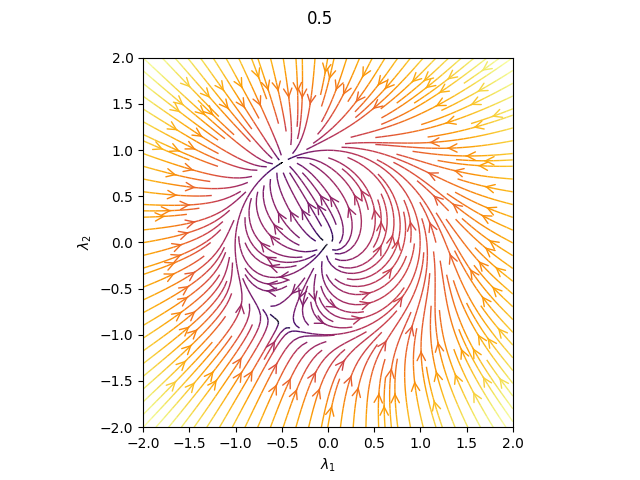}
        \caption{$\mu=0.5$}
    \end{subfigure}
    \begin{subfigure}[b]{0.6\textwidth}
        \includegraphics[width=\textwidth]{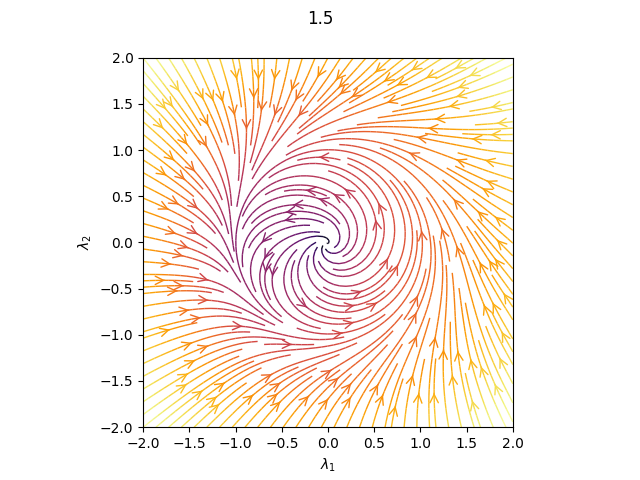}
        \caption{$\mu=1.5$}
    \end{subfigure}
    \caption{Example of a SNIPER Bifurcation. The fixed points move along the circle with radius 1 until they meet and annihilate, leaving a limit cycle in its place.}
    \label{ASDKJOASUBDVASTGIDHOJIAPSKDASNHDGYTASGYHJKM}
\end{figure}


\newpage

\subsection{Marginality Crossing}\label{MArg}

Many types of bifurcations seen so far display a common feature: one of the fixed points becomes marginal at the fixed point, i.e., one of the eigenvalues of the Jacobian of $\beta$ vanishes. In terms of CFT's, we say that the operator related to the eigendirection that has a vanishing eigenvalue crosses through marginality. Marginality crossing is one of the main reasons \cite{GukovBif} considers the study of bifurcations in RG flows. Taking the saddle-node bifurcation (\ref{saddle}) as an example, the Jacobian is just the derivative of $\beta$:

\begin{equation}
    J = \frac{\partial\beta}{\partial \lambda}=  - 2 \lambda .
\end{equation}

We see that for the fixed points $\lambda_\pm = \pm \sqrt{r}$ the Jacobian goes to zero as the control parameter approaches the bifurcation point $r=0$. This means, in the field theory, that the operator related to $\lambda$ that was irrelevant for $r>0$ becomes exactly marginal when the fixed points merge. In dynamical systems, as seen in section \ref{Linear}, marginal fixed points can take various forms in the phase space. For the case of the saddle-node bifurcation, the fixed point at the bifurcation becomes a semi-stable fixed point, i.e., from one direction of the flow the fixed point is stable and for the other direction the fixed point is unstable. See figure \ref{ACOSKOKDAOKdo}.

\begin{figure}[h!]
\centering
\includegraphics[width=0.7\linewidth]{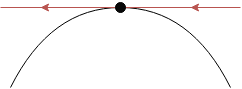}
\caption{When $r=0$ the operator becomes marginal and dynamically the fixed point is said to be semi-stable. The black line shows $\beta(r)$.}
\label{ACOSKOKDAOKdo}
\end{figure}

Also, Jacobians of different bifurcations have different rates of vanishing. While for the saddle-node this rate scales like $J \alpha \sqrt{r}$, for transcritical and pitchfork bifurcations the rate is $J\alpha r$. So, if we know the relation of the conformal dimension (remember that $eigenvalue(J)_i=\Delta_i-d$) of an operator with respect to the control parameter, we can identify the possible bifurcation in the flow that will occur.

\newpage


\section{RG Flows as Dynamical Systems}

\subsection{Reconstructing the theory space from RG flows}\label{dasdasdasdxzc}

This subsection is in part based on section 2 of \cite{GukovCount}. In theory, we could construct the space $\Tau$ of all quantum field theories that preserve the symmetries as the RG flow we are studying. Such constructions may lead to infinite dimensional theories, but we may have the case that not all dimensions are ``relevant" for the flow.

Without losing generality, let us say we are studying the flow between two fixed points $\lambda_{UV}$ and $\lambda_{IR}$. We define an index that counts the number of relevant eigendirections, i.e., the number of relevant operators:
\begin{equation}\label{MorseIndex}
\mu(\lambda)=\#(\text{relevant spin-0 } \mathcal{O}).
\end{equation}
Then we can define
\begin{equation}
\mathcal{I}(\lambda_{IR})\coloneqq \left\lbrace \lambda_t \in \Tau | \lim_{t\rightarrow \infty} \lambda_t = \lambda_{IR} \right\rbrace,
\end{equation}
as the irrelevant manifold of $\lambda_{IR}$ and
\begin{equation}
\mathcal{R}(\lambda_{UV})\coloneqq \left\lbrace \lambda_t \in \Tau | \lim_{t\rightarrow - \infty} \lambda_t = \lambda_{UV} \right\rbrace,
\end{equation}
as the relevant manifold of $\lambda_{UV}$.

Now we can define the space of RG flows between the fixed points, called moduli space, as
\begin{equation}
\mathcal{M}(\lambda_{UV},\lambda_{IR})=\mathcal{R}(\lambda_{UV})\cup\mathcal{I}(\lambda_{IR}).
\end{equation}
The moduli space is usually finite-dimensional since, if $\mathcal{R}(\lambda_{UV})$ and $\mathcal{I}(\lambda_{IR})$ intersect tranversely,
\begin{equation}
\text{dim } \mathcal{M}(\lambda_{UV},\lambda_{IR}) = \mu(\lambda_{UV}) - \mu(\lambda_{IR}).
\end{equation}
In other terms, the dimension of the space of flows is the number of operators that were relevant in the UV and are not relevant in the IR.


\subsection{C-Theorem}

In 1986, Zamolodchikov \cite{ZamOG} achieved a breakthrough in the study of RG flows. He proved the following theorem

\begin{theorem}\textbf{Zamolodchikov's c-Theorem}\\
For a 2d field theory we have the following
\begin{enumerate}
    \item There exists a function $c(\lambda )\geq 0$ such that 
    $$\dfrac{d}{dt}c = \beta_i \dfrac{\del}{\del \lambda_i} c (\lambda) \leq 0, $$
    where the equality is only true at the fixed points of $\beta$.
    \item Fixed points are stationary for $c$, namely
    $$\beta_i(\lambda^*) = 0 \rightarrow \dfrac{\del c}{\del \lambda^*}=0.$$ Also fixed points are conformal field theories with central charge $\tilde{c}$ that appears in the Virasoro Algebra.
    $$[L_m,L_n] = (m-n)L_{m+n}+\frac{\tilde{c}}{12}(m^3-m)\delta_{m+n.0}.$$
    \item The value of $c(\lambda)$ at a fixed point $\lambda^*$ is equal to the central charge of the CFT:
    $$c(\lambda^*)=\tilde{c}(\lambda^*).$$
\end{enumerate}
\end{theorem}

As we can see, the $c$-function resembles the Lyapunov function of systems like those seen in sectiob \ref{Lyplyp}. As it was the case for potential functions (and for Lyapunov functions as well), there is no general method for finding $c$-functions or proving that they do not exist. Usually, to prove that such a function exists, one needs to actually find one using nothing more than intuition. 

In conformal field theories in 2d, we have the fact that there is an energy-momentum tensor $T_{\mu\nu}$ that satisfies $\del_\mu T_{\mu\nu}=0$. We also define  $z =x^1+ix^2$ and $T=T_{zz}$, $\Theta = T_{z\bar{z}}$. The latter definition can be expended in the local fields $\Phi_i(x)$

\begin{equation}
    \Theta = \sum_i\beta_i(\lambda)\Phi_i(x).
\end{equation}

To prove the theorem above, Zamolodchikov defines the auxiliary correlation functions
\begin{align}
    C(\lambda) = & \left. 2z^4\left\langle T(x)T(0)\right\rangle \right|_{x^2=x^2_0},\\
    H_i(\lambda) = & \left. z^2x^2\left\langle T(x)\Phi_i(0)\right\rangle \right|_{x^2=x^2_0},\\
    G_{ij}(\lambda) = & \left. x^4\left\langle\Phi_j(x)\Phi_i(0)\right\rangle \right|_{x^2=x^2_0}.
\end{align}
where $x_0$ is an arbitrary scale set to 1. Through various manipulations of the Callan-Symanzik equations we arrive at
\begin{align}\label{essasporra}
    \frac{1}{2}\beta^i\del_i C = & -3\beta^iH_i +\beta^i\beta^k\del_kH_i + \beta^k(\del_\beta^i)H^i, \nonumber\\
    \beta^k\del_k H_i + (\del_i\beta^k)H_k - H_i = & -2\beta^kG_{ij} +\beta^j\beta^kG_{ij}+\beta^j(\del_i\beta^k)G_{jk}+\beta^j(\del_j\beta^k)G_{jk}.
\end{align}
To fill in the details \cite{Naga} is a good resource. Now we define the following function

\begin{equation}
    c(\lambda) = C(\lambda ) + 4\beta^iH_i(\lambda ) - 6\beta^i\beta^jG_{ij}(\lambda ).
\end{equation}
Plugging in (\ref{essasporra}) we find

\begin{equation}
    \beta^i\del_i c = -12 \beta^i\beta^jG_{ij},
\end{equation}
where $G_ij$ is positive definite \cite{ZamOG} and therefore $c$ obeys the relations in assertion 1 of Zamolodchikov's $c$-Theorem. The prove of the other two assertions can be found in \cite{ZamOG}. 

The existence of a $C$-function has a tremendous impact in our understanding of RG flows. As a start, it gives an idea of a ``natural order'' of flows, where the function is always decreasing. Also since the central charge measures the number of degrees of freedom of a CFT, we see that this natural order is in the direction of the loss of information. This is the idea given by Zamolodchikov of irreversibility, starting from a theory with more degrees of freedom (the theory in the UV) the flow will lead us to a theory with less degrees of freedom (the theory in the IR), see figure \ref{asdaodmvcxnhskdfhjjsdkgbvfhnsho}.

\begin{figure}[h]
\begin{center}
\includegraphics[width=0.65\linewidth]{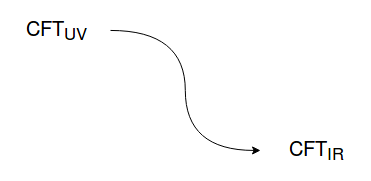}
\caption{An example of an ordinary RG flow between two fixed points. The order of the flow is induced by the diminishing of the $c$-function from an UV fixed point to an IR fixed point.}
\label{asdaodmvcxnhskdfhjjsdkgbvfhnsho}
\end{center}
\end{figure}

Also, the $c$-function gives us a way to make an analogy with RG flows and thermodynamics. Thermodynamical systems have a quantity called entropy, and entropy always increase or stays the same with time, the reverse of the $c$-function. In fact, a quantity called \textit{entanglement entropy} can be defined in field theories and we can use this to construct $c$-functions of RG flows.

For dimension higher than 2, alternative theorems were constructed. For example, for 3 dimension we have the $F$-theorem and for 4 dimensions we have the $a$-theorem. Since the difference between these theorems are not the focus of this thesis, we are going to collectively call all of these theorems as $C$-theorems. We can divide these theorems in three versions \cite{BIWW}.

\begin{enumerate}
\item Weak Version:

There exists a function $C:\Tau\rightarrow\Re$ such that $C(\lambda_{UV})>C(\lambda_{IR})$.

\item Strong Version:

$C$ is decreasing along the flow: 
\begin{equation}
\frac{dC}{dt}(\lambda)\leq0,
\end{equation}
and if $\beta(\lambda^*)=0$, then $\dfrac{dC}{dt}(\lambda^*)=0$.

\item Strongest Version:

The flow is a gradient system
\begin{equation}
\beta(\lambda)=\nabla C(\lambda).
\end{equation}
\end{enumerate}

The strongest version is enough to rule out a number of nonlinear behaviours like limit cycles. The strong version is very similar to the definition of the Lyapunov function, with a small relaxation: outside the fixed point we have $\dot{C}\leq 0$, while the Lyapunov function $V(\lambda)$ assumes $\dot{V}<0$. This relaxation is enough to allow the possibility of RG flows where the strong version of the $C$-Theorem is valid and still one could find limit cycles.

The weak and strong forms of the $C$-theorem have been proven in two \cite{ZamOG} and four \cite{Komargodski} dimensional QFT's. The strongest version is the hardest to prove, despite the fact that it has compelling non-perturbative arguments for a wide range of RG flows in two dimensions \cite{RFLM} and in four \cite{BIWW},\cite{AGPZ}. In fact, \cite{GradFormula} provides a general gradient formula with relative few assumptions.

There are however counter-examples, \cite{FGCDec} found a $d=4$ unitary QFT (as well as other examples in $d=4-\epsilon$) that exhibits a limit cycle in the theory space. One interesting fact that they discovered in \cite{FGCNov} is that conformal symmetry is preserved around the orbit, giving an example of a CFT where $\beta = 0$, a fact that was suggested earlier by \cite{JO1990}.
.

\subsection{RG cycles}

To see how RG flows with limit cycles can appear we follow the constrution made by Jack and Osborn in \cite{JO1990} following the notation in \cite{JO2000} and \cite{FGCSusy}. 

They first introduce the ideia to consider the coupling constants $\lambda$ as dependent of the spacetime coordinates. This is usefull to find expressions of quantities like the stress-energy tensor in terms of a functional derivative of the generating functional $W$ in the form

\begin{equation}
\left\langle T_{\mu\nu}(x)\right\rangle=\frac{2}{\sqrt{-g}}\frac{\delta W}{\delta \lambda^{\mu\nu}(x)},
\end{equation}

\begin{equation}
\left\langle \left[ \mathcal{O}_i(x) \right]\right\rangle=\frac{1}{\sqrt{-g}}\frac{\delta W}{\delta \lambda^{i}(x)}.
\end{equation}

Equivalent constructions were made by Komargodski \cite{KoCons} and Luti, Polchinski and Rattazzi \cite{LPR} in order to prove the $a$-theorem, the analog to the $c$-theorem for four-dimensional QFT`s, and to find the assymptotic behaviour of four-dimensional RG-Flows. In any case, the coupling's dependence on space-time will lead us to new counterterms. 

Let us consider the following unnormalized Lagrangian

\begin{equation}
\mathcal{L} = \frac{1}{2} \partial^\mu\ \phi_{0a} \partial_\nu \phi_{0a} - \frac{1}{4!}\lambda^0_{abcd}\phi_{0a}\phi_{0b}\phi_{0c}\phi_{0d},
\end{equation}
and the right transformations for the Lagrangian to be symmetric under a flavor symmetry $SO(n)$. One interesting counter term to make this Lagrangian finite is
\begin{equation}
\mathcal{L}_{ct}=\partial^\mu \lambda_I(N_I)_{ab}\phi_{0b}\partial_\mu \phi_{0a},
\end{equation}
where $I=\left\lbrace abcd \right\rbrace$ and $N^I$ is an element of the Lie algebra of $SO(n)$.

We can introduce gauge fields by substituting $\partial_\mu$ with $D_\mu=\partial_\mu+A_{0\mu}$, and their renormalizations are
\begin{equation}
A_\mu=A_{0\mu} - N_I(D_\mu \lambda)_I.
\end{equation}

Using this together with the equations of motions, Jack and Osborn get the equation

\begin{equation}\label{porracacete}
T^\mu_\mu = (\beta_I-(S\lambda)_I)\left[\mathcal{O_I}\right] - ((1+\gamma+S)\phi)\frac{\delta}{\delta \phi}S_0,
\end{equation}
where $S=-\lambda_IN_I^1$, $N_I^1$ is the residue of the pole in $N_I$.

From the equation and the fact that conformal theories are characterized by $T^\mu_\mu=0$, we have CFT's when $\beta = S\lambda$. The cases where $S=0$ are related to fixed points, and in a different direction \cite{FGCNov} found that S is zero in fixed points. The complex case happens when $S\neq 0$, and conformal theories are of the form $\beta=Q\lambda$, where $Q$ is an element of an Lie algebra, and this is the case where cycles appear.

If we define

\begin{equation}
    B_I = \beta_I - (S\lambda )_I,
\end{equation}
equation (\ref{porracacete}) resembles an equation without the possibility of RG cycles for a dynamical system given by $B_I$:

\begin{equation}
T^\mu_\mu = B_I\left[\mathcal{O_I}\right] - ((1+\gamma+S)\phi)\frac{\delta}{\delta \phi}S_0.
\end{equation}

One question we may ask: How does the $c$-function behaves in the presence of a limit cycle? In \cite{FGCNov} the authors compute that the function

\begin{equation}
    \tilde{B} = a +\frac{1}{8}B_Iw_I,
\end{equation}
where $a$ is the trace anomaly in 4 dimensions and $w$ is a function of the coupling that comes from consistency conditions\cite{FGCNov}. This function has the property of a weak $c$-function

\begin{equation}
    \tilde{B}_{UV} > \tilde{B}_{IR},
\end{equation}
but \textbf{in} the limit cycle we can not have any monotonically decreasing function of the couplings. This is obvious since periodic flows will pass through the same point multiple times as the flow progresses, and after 1 revolution done in time $t$ we have

\begin{equation}
    \beta (0) = \beta (t),
\end{equation}
and therefore we need to have that

\begin{equation}
    c(\beta(0)) = c(\beta(t)).
\end{equation}
Since the $c$-function is never increasing, the only possibility we have is that the $c$-function is constant around the cycle.

In our analogy with thermodynamics, cycles resemble isentropic process, i.e., a process where the entropy is constant. Isentropic process are reversible and adiabatic at the same time and the isentropic flows in fluid dynamics have no energy added to it and there is no dissipative effects like friction along the flow.

\begin{figure}[h]
\begin{center}
\includegraphics[width=0.5\linewidth]{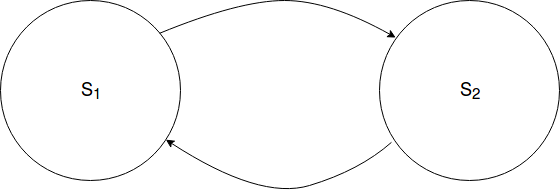}
\caption{A reversible process in thermodynamics. If the cycle is adiabatic, entropy doesn't change and the process is said to be isentropic.}
\end{center}
\end{figure}

\subsubsection{Cycles as ambiguities in RG flows}

In \cite{LPR} the authors revise the limit cycles above and they understand then as ambiguities in the Callan-Symanzik equation that come from the flavor symmetries. To see this, we first start with the renormalized Lagrangian of a theory with $N$ scalar fields.

\begin{equation}
    \mathcal{L}_R = \frac{1}{2}\del\phi^i\del\phi^i - \frac{\lambda_{ijkl}}{4!}\phi^i\phi^j\phi^k\phi^l.
\end{equation}

The Callan-Symanzik equations staes that correlation functions are independent of a rescaling of the fields. We can see that by defining a rescaling like

\begin{equation}
    \phi^{i\prime} = \xi^i_a\phi^a.
\end{equation}

With this the Callan-Symanzik equation becomes \cite{LPR}:

\begin{equation}
\left(\frac{\partial}{\partial t}+\beta_i\frac{\partial}{\partial\lambda_i}\right)<\phi^{i_1}...\phi^{i_n}>= \gamma^{i_1}_k<\phi^{k}...\phi^{i_n}>+...+\gamma^{i_n}_k<\phi^{i_1}...\phi^{k}>
\end{equation}
where 
\begin{equation}
    \gamma^i_j\xi^j_a(\lambda) = \dfrac{d}{dt}\xi_i^a.
\end{equation}

We are going to consider that the flow is equivalent to a rotation

\begin{equation}
    \lambda (t) = R(t)\bar{\lambda}.
\end{equation}
From the fact that $R$ is a rotation we have that $R^TR=\mathbb{I}$. Then by defining the transformation
\begin{equation}
\phi^{i\prime} = R(t)^i_ j\phi^{j}
\end{equation}
we have that the Lagrangian becomes

\begin{equation}
\mathcal{L}_R = \frac{1}{2}\del\phi^{i\prime}\del\phi^{i\prime} - \frac{\bar{\lambda}_{ijkl}}{4!}\phi^{i\prime}\phi^{j\prime}\phi^{k\prime}\phi^{l\prime}.
\end{equation}

For this Lagrangian $\bar{\lambda}$ does not depend on the RG-time and so this is a fixed point. We can also see the effects of this rescaling on the S matrix definded before in (\ref{porracacete})

\begin{equation}
    S_I = \lambda_IN_I^1.
\end{equation}

We now assume $R$ can be expanded in poles, $R = 1 + R_1/\epsilon + ...$. The effect of the transformation on $N_I^1$ can be calculated to be \cite{FGCNov}
\begin{equation}
    N_I^1 \rightarrow N_I^1 - \del_IR_1,
\end{equation}
so the effect of $S$ is

\begin{equation}
    S \rightarrow S + \lambda_I\del_IR_1.
\end{equation}

We can find a gauge such that $\lambda_I\del_IR_1 = - S$. In this case, $B_I = \beta_I$ and the limit cycles are rescaled to fixed points. 

\subsection{Multi-Valued C-functions}

Curtright \textit{et al.} in \cite{Folklore} make an argument against theorem \ref{Teo1}. They claim they have found an example of a gradient flow with a periodic trajectory.

The example is the following RG flow with only one relevant eigendirection:
\begin{equation}\label{FolkModel}
\beta=\pm\sqrt{1-\lambda^2}.
\end{equation}
This flow is equivalent with the ``Russian doll Superconductivity Model". The solution for this model described by the authors resembles a simple harmonic oscillator, the positive branch is an oscillator with a ``right movement" until it hits $\lambda=1$, then the flow moves to the negative branch and is left moving until it hits $\lambda=-1$, where the cycle repeats itself.

The $C(\lambda)$ function for the system is
\begin{equation}
C_N=-\frac{\pi}{4}(1+2N)-(-1)^N\left(\frac{1}{2}\arcsin(\lambda)+\frac{1}{2}\lambda\sqrt{1-\lambda^2}\right),
\end{equation}
and $N$ represents the number of times the flow has turned. We have
\begin{equation}
\beta(\lambda)=\frac{d\lambda}{dt}=(-1)^N\sqrt{1-\lambda^2} = - \frac{d}{d\lambda}C_N(\lambda),
\end{equation}
and $C$ is monotonically decreasing along the flow
\begin{equation}
C(t)=-\frac{1}{2}(t-\cos(t)\sin(t)) \quad,\quad \frac{dC}{dt}\leq 0.
\end{equation}
Again, we see that the strong version of the $C$-theorem is not incompatible with periodic trajectories, the strongest version is. For this model, the strongest version is not valid since the $C$-function fails to be single-valued in theory space, and so the flow as described by the authors is not a gradient flow as it was defined earlier. 

In fact, we know for sure that vector fields on the line do not oscillate, since the solutions will always be monotonic and thus the turning points at $\lambda = \pm1$ are impossible. Since RG flows are a system of first order differential equations, a system with only one relevant operator for the entire theory space can not have periodic flows.

What \cite{Folklore} may suggest us is a relaxation of the strongest version of the $C$-theorem which allows $C$ to be multi-valued. This version could be useful when studying systems such as the one presented above.


\newpage

\newpage




\chapter{Examples}

In this chapter we review some examples of the ideas encountered so far.

\section{$O(n)$ model}

Given $n$ scalar fields $\phi_i$ organized as
\begin{equation}
    \Phi = (\phi_1,...,\phi_n),
\end{equation}
we consider Lagrangians which are invariant under the transformation
\begin{equation}
    \Phi_a(x) = R^b_a\Phi_b,
\end{equation}
where $R$ is an element of the $O(n)$ group, i.e., it is an operator such that $R^TR=\mathbb{I}$. A general action that is invariant under $O(n)$ is of the form:

\begin{equation}
    S [\Phi ] = \int d^d x \left[ \nabla \Phi \cdot \nabla \Phi + t \Phi \cdot \Phi + u\sum_{ij}\phi_i^2\Phi_j^2 + ... \right].
\end{equation}
This theories are known as \textit{$O(n) models$}. They have important applications in statistical field theory. For example, the $O(2)$ model (also known as the XY-model) describes magnets such that the spins can rotate in a plane. It is convienient for this case to rewrite the fields as:
\begin{equation}
    \psi(x) = \phi_1(x) + i\phi_2(x).
\end{equation}
Now the action
\begin{equation}
    S [\Phi ] = \int d^d x \left[ \nabla \bar{\psi} \cdot \nabla \psi + t \bar{\psi}\psi + u\sum_{ij}(\bar{\psi}\psi)^2 \right],
\end{equation}
becomes invariant under $U(1)$ transformations $\psi \rightarrow e^{i\alpha}\psi$. Another physical system that can be described by this is a Bose-Einstein condensate $\cite{Tong}$ where $\psi$ represents ``off-diagonal long range order in one-particle density matrix''.

The $O(3)$ model (also known as the \textit{Heinsenberg model}) describes spins that can rotate in a three dimensional space.

\subsubsection{The $\epsilon$ expansion}

For $d=2$ the theory has no interacting critical points with $O(n)$ symmetry \cite{Tong}. In $d=3$ the interacting critical point depends on $n$. We can see $n$ as a control parameter in the RG-equations. We are more interested in the theory in 
\begin{equation}
    d=4-\epsilon
\end{equation}
dimensions. The RG equations can be calculated to be, at a leading order in $\epsilon$ and $u$ \cite{ON(N)}:
\begin{align}
\begin{split}
    \beta_u = & \epsilon u - -8(n+8)u^2 , \\
    \beta_t = & 2t - 8(n+2)ut .
\end{split}
\end{align}

There are two fixed points, the trivial one at $u=0$ and one at 
\begin{equation}
    u^*=\frac{\epsilon}{8(n+8)}.
\end{equation}. 
We see that in the limit $n\rightarrow\infty$, $u^*\rightarrow 0$ and the theory is transformed into a theory with only one trivial fixed point. The flow for $t$ has only one critical value at $t=0$.

\subsubsection{Cubic Symmetry}

The RG flow becomes more interesting in the case that we add a symmetry breaking form the $O(n)$ group to its cubic subgroup. In the examples of XY and Heinsenberg models before, this \textit{cubic symmetry breaking} describes when the magnets are arranged in a crystal of cubic symmetry. Near four dimensions, the only relevant term to be added to the $O(n)$ model is of the form
\begin{equation}
    v\sum_i\phi_i^4,
\end{equation}
and so the actions becomes

\begin{equation}
    S [\Phi ] = \int d^d x \left[ \nabla \Phi \cdot \nabla \Phi + t \Phi \cdot \Phi + u\sum_{ij}\phi_i^2\Phi_j^2 + v\sum_i\phi_i^4 \right].
\end{equation}

This model is used as the introductory example in \cite{GukovBif} to study bifurcations in RG flows. The beta functions are calculated in \cite{ON(N)} and they are, up to quadratic order, given by
\begin{align}\label{ONONON}
\begin{split}
    \beta_u = & \epsilon u - 8(n+8)u^2 - 48 uv , \\
    \beta_v = & \epsilon v - 96uv - 72v^2 , \\
    \beta_t = & 2t - 8(n+2)ut - 24vt .
\end{split}
\end{align}

Both $\beta_u$ and $\beta_v$ are independent of $t$, so we can analyse the RG flow in the $uv$-plane alone. This system has four fixed points. They are 

\begin{itemize}
    \item Gaussian (G): it corresponds to the trivial fixed point at $(u,v)=(0,0)$,
    \item Wilson-Fisher (H): it is the non-trivial fixed point when $v=0$, \\ located at $(\frac{\epsilon}{8(n+8)},0)$,
    \item Ising (I): it is the non-trivial fixed point when $u=0$, located at $(0,\frac{\epsilon}{72})$,
    \item Cubic (C): it is located at $(\frac{\epsilon}{24n},\frac{(n-4)\epsilon}{72n})$.
\end{itemize}

The jacobian of this system is:

\begin{equation}
    J = 
    \begin{bmatrix}
    \frac{\partial \beta_u}{\partial u}       & \frac{\partial \beta_u}{\partial v} \\
    \frac{\partial \beta_v}{\partial u}       & \frac{\partial \beta_v}{\partial v} 
    \end{bmatrix}
    = 
    \begin{bmatrix}
    \epsilon - 16(n+8)u - 48v       & -48u \\
    -96v      & \epsilon - 96u - 144v
    \end{bmatrix}
    \label{coolbeans}
\end{equation}
Plugging in $u$ and $v$ for each of the fixed points gives us the eigenvalues, correspondingly given in table \ref{Table}.

\begin{table}[h!]
    \centering
    \begin{tabular}{c|cc}
         Fixed point & \multicolumn{2}{|c}{Eigenvalues} \\
  \hline\\
         G & $\epsilon$ &$\epsilon$ \\
         H & $-\epsilon$ &$\epsilon\dfrac{n-4}{n+8}$ \\
         I & $-\epsilon$ &$\dfrac{\epsilon}{e}$ \\
         C & $-\epsilon\dfrac{n-4}{3n}$ & $-\epsilon$ \\
    \end{tabular}
    \caption{The stabilities of the fixed points can be derived from the eigenvalues of the matrix (\ref{coolbeans})}
    \label{Table}
\end{table}

We can see that the Gaussian fixed point is completely unstable and the Ising fixed point is a saddle independently of $n$. But the Wilson-Fisher and the Cubic do depend of $n$, the fixed points exchange stability at a critical value of $n_c=4$, below the critical value H is a saddle and C is an stable fixed point, above the critical value H becomes stable and C becomes a saddle. Also, if we look at the positions of these points, we see that they meet when $n=4$. Exactly at the critical value, both fixed points have zero eigenvalues, this shows that the operators crossing though marginality. This is the perfect description of a transcritical bifurcation, studied in the previous chapter. See figure \ref{ONfigures}.

\begin{figure}
    \centering
    \begin{subfigure}[b]{0.6\textwidth}
        \includegraphics[width=\textwidth]{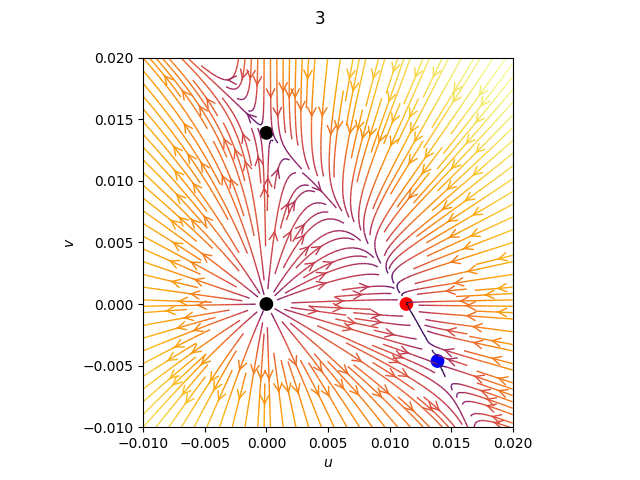}
        \caption{$n=3$}
    \end{subfigure}
    ~ 
    \begin{subfigure}[b]{0.6\textwidth}
        \includegraphics[width=\textwidth]{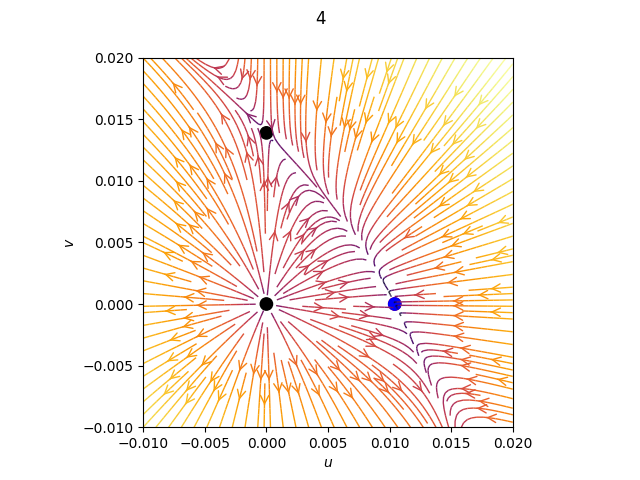}
        \caption{$n=4$}
    \end{subfigure}
    \begin{subfigure}[b]{0.6\textwidth}
        \includegraphics[width=\textwidth]{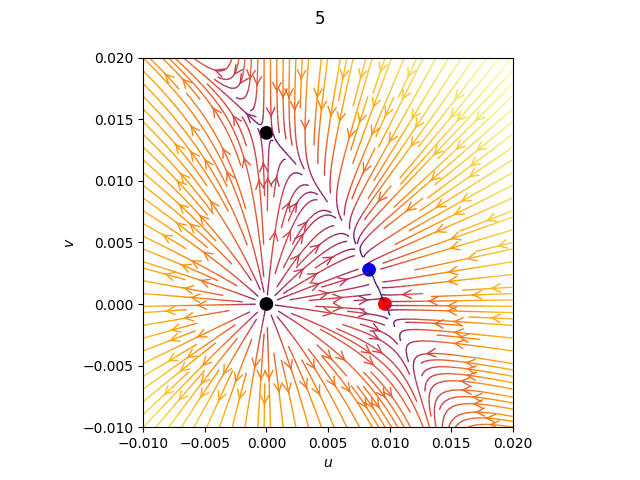}
        \caption{$n=5$}
    \end{subfigure}
    \caption{Phase space of the system (\ref{ONONON}) for different values of n. Notice how the Wilson-Fischer fixed point collides with the cubic fixed point at the critical value $n=4$, resulting in a transcritical bifurcation.}\label{ONfigures}
\end{figure}

We also have the fact that in the limit of $n\rightarrow\infty$, the Cubic fixed point merges with the Ising fixed point, so in the two regimes $n=n_c$ and large $n$ the system has only three fixed points.

Now this was calculated using perturbation theory up to quadratic order, one question we may ask is if the system continues to have this behaviour if we include more orders in perturbation theory. As we have seen in the previous chapter, the transcritical bifurcation is of codimension 2, this means that for a system with only one control parameter small deviations lead to the unfolding of the bifurcation. However, we already explored the possible unfoldings of the transcritical bifurcation in section \ref{tranun}. The first possibility is that the fixed points "miss" each other and the bifurcation does not occur. We can rule out this in the $O(n)$ case since we know that the cubic and the Wilson-Fisher do exchange their stabilities. The other possibility is that there are two saddle-node bifurcations in the place of the original transcritical. As we have seen in section \ref{MArg}, if we could find the expression of how the conformal dimension of the fixed points depends on $n$ we can determine if the bifurcation unfolds or not.

\subsubsection{Index theories}

We are going to show how we can use both indices seen in the previous chapter to prove that there is no limit cycles in the model. Of course, by the diagrams in figure \ref{ONfigures} we can see that there is no periodic flow on the theory space, but is not always that the diagrams are so conclusive.

Without losing generality, we will focus on the model with $n>n_c$, for different values of $n$ the argument would be very similar. First we will use the Conley index to show that the fixed points in the system are not isolated. Let us define the set $N$ as a set that encloses all fixed points in the system. Since the exit set for $N$ is only one isolated interval, the Conley Index is

\begin{equation}
    CH_*(\text{Inv} N) = 0.
\end{equation}

So if the fixed points were isolated, by properties of the Conley Index we should have that the sum of the indices of the points to be equal to zero. Clearly, that is not what happens, in fact

\begin{equation}
    CH_*(C \cup I \cup H \cup G) = (\mathbb{Z},\mathbb{Z}\oplus \mathbb{Z},\mathbb{Z},0...).
\end{equation}

\begin{figure}[h!]
    \centering
    \includegraphics[width=0.6\textwidth]{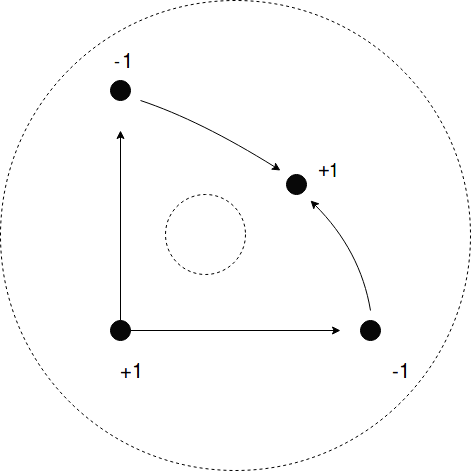}
    \caption{The only possible closed curves that do not intersect flow lines between the fixed points are the ones that encloses none or all the fixed points.}
    \label{ncaidasc}
\end{figure}

We can refine this analysis. If we set $N$ as a set that encloses any pair of fixed points (and just the pair) the Conley index of $N$ would be $0$, but the Conley Index of the fixed points in the pair do not sum up to zero. This shows that there are flow lines between the fixed points. The only exception is the pair $(H,I)$, and indeed we see that these fixed points are not connected.

Now we try to find periodic orbits in theory space. As seen in section \ref{inccloscur}, the index of a closed curve of a limit cycle is equal to $+1$. Also, the theorem of uniqueness of the solutions \ref{Intersecttheorem} says that we cannot have intersection in the diagram. Since we know that there are flow lines between the fixed points, the only closed curves we can form that do not intersect the flux, are the ones that enclose no fixed points or that enclose all fixed points. We can easily rule out the former curves, since their index is $0$. And for the latter curves, we can use theorem \ref{coolctheorem} to calculated that the index is

\begin{equation}
    I = I_C + I_H + I_I + I_G = +1 -1 -1 +1 = 0.
\end{equation}
Thus we also see that these curves do not represent periodic flows of the system. See figure \ref{ncaidasc}.



\section{4-dimensional QCD}

The four dimensional quantum chromo-dynamics (QCD) is one of the most interesting examples of field theories. It consists of a $SU(N_C)$ gauge theory with $N_f$ massless flavors. The particularities of QCD lead to interesting behaviours that are yet to be fully comprehended, like \textit{color confinement}. We can however infer a lot studying the renormalization group flows from these theories. 

The beta function at one-loop for a non-Abelian gauge theory is

\begin{equation}
    \beta = -\frac{-g^3}{(4\pi)^2}\left[\frac{11}{3}C_2(G)-\frac{4}{3}N_fC(R)\right],
\end{equation}
where $g$ is the gauge coupling, $N_f$ the number of flavor, $C(R)$ is a constant depending on the fermion representation $R$ and $C_2(G)$ is another constant called the quadratic Casimir operator of the adjoint representation of the group. For $SU(N_C)$, we have $C(r)=1/2$ and $C_2(G)=N_C$, then the beta function becomes \cite{Peskin}

\begin{equation}
    \beta = -\frac{-g^3}{(4\pi)^2}\left[\frac{11}{3}N_C - \frac{2}{3}N_f\right].
\end{equation}
There is only one fixed point for this $\beta$, the trivial fixed point. Then, for big enough $N_C$, the gauge coupling will always run down to zero, leading to a free (and conformal) theory. This phenomena is known as \textit{asymptotic freedom} and is one of the most interesting behaviours in field theories.

As we increase the number of loops, new fixed points may appear. The general beta function for $\alpha_g = \propto g^{3/2}$ is

\begin{equation}
    \beta = \gamma \alpha_g - b_0 \alpha_g^2 - b_1 \alpha_g^3  + ... \hspace{10 pt} .
\end{equation}
For the QCD, the constants are
\begin{align}
    \begin{split}
        \gamma = & 0,\\
        b_0 = & (11 - 2r)\frac{2N_C}{3},\\
        b_1 = & (34-13r+\frac{3r}{N_C^2})\frac{2N_C^2}{3},\\
        ... & 
    \end{split}
\end{align}
where we defined the ratio 
\begin{equation}\label{oipo}
    r=\frac{N_f}{N_C}.
\end{equation}
Since there are two control parameters, $r$ and $N_C$, we can represent the parameters as $(b_0,b_1)$, and all $b_n$, $n\pm 2$, can be written in terms of $b_0$ and $b_1$. Also, we can expect up to codimension-2 bifurcations to be stable. In fact, if we consider the two loop beta function we can see that there is another fixed point in the system other than the trivial, $\alpha^*=-\frac{b_0}{b_1}$. Without loss of generality, let us fix $b_0 > 0$ and we see that as $b_1$ changes sign we have a transcritical bifurcation in the system: the trivial fixed point and $\alpha^*$ exchange stability. When $b_1$ is positive the trivial fixed point is stable and the system shows asymptotic freedom, but when $b_1$ is negative then $\alpha^*$ is stable and the flow asymptotically goes to an interacting conformal field theory. There is one problem: $\alpha^*$ has a pole at the bifurcation point. This may suggest that we need to include flavor dynamics to the theory in order to see the bifurcation. 

It is a popular suggestion that at the critical value of $N_C$ there is an annihilation of fixed points leading to the loss of conformality of the flow \cite{ConfLoss}. This suggests a saddle-node bifurcation at that value. \cite{QCDSADLE} provides a toy model for QCD considering the 4-fermi operator added to the gauge action:

\begin{equation}
    S = \int dx^4 \left(  -\frac{1}{4g^2}F^A_{\mu\nu}F^{A\hspace{2 pt}\mu\nu} + i \bar{\psi}_f\slashed{D}\psi^f + \mathcal{L}_{4f} \right),
\end{equation}
with the 4 fermi Lagrangian being, using the notation of \cite{GukovBif}, 
\begin{equation}
    \mathcal{L}_{4f} = \frac{\lambda_1}{4\pi^2\Lambda^2}\mathcal{O}_1 + \frac{\lambda_2}{4\pi^2\Lambda^2}\mathcal{O}_2 + \frac{\lambda_3}{4\pi^2\Lambda^2}\mathcal{O}_3 + \frac{\lambda_4}{4\pi^2\Lambda^2}\mathcal{O}_4,
\end{equation}
and
\begin{align}
    \begin{split}
        \mathcal{O}_1 = &  \bar{\psi}_i\gamma^\mu\psi^j\bar{\psi}_j\gamma_\mu\psi^i + \bar{\psi}_i\gamma^\mu\gamma_5\psi^j\bar{\psi}_j\gamma_\mu\gamma_5\psi^i\\
        \mathcal{O}_2 = &  \bar{\psi}_i\psi^j\bar{\psi}_j\psi^i + \bar{\psi}_i\gamma_5\psi^j\bar{\psi}_j\gamma_5\psi^i\\
        \mathcal{O}_3 = & (\bar{\psi}_i\gamma^\mu\psi^i)^2 - (\bar{\psi}_i\gamma^\mu\gamma_5\psi^i)^2 \\
        \mathcal{O}_3 = & (\bar{\psi}_i\gamma^\mu\psi^i)^2 + (\bar{\psi}_i\gamma^\mu\gamma_5\psi^i)^2.
    \end{split}
\end{align}
The indices $i$ and $j$ refer to the flavor symmetries. In the notation of \cite{QCDSADLE}, $\mathcal(O)_1 = 2\mathcal{O}_V$, $\mathcal(O)_2 = 2\mathcal{O}_S$, $\mathcal(O)_3 = 2\mathcal{O}_{V1}$  and $\mathcal(O)_4 = 2\mathcal{O}_{V2}$. The RG flow equations for these couplings are 
\begin{align}
    \begin{split}
        \beta_1 = & 2\lambda_1 + \frac{N_f}{4}\lambda^2_2+ (N_C+N_f)\lambda^2_1 - \\
        & - 6\lambda_1\lambda_4 - \frac{6}{N_C}(\lambda_1+\lambda_4)\alpha - \frac{3}{4}\left(N_C-\frac{8}{N_C}+\frac{3}{N_C^2}\right) \alpha^2,\\
        \beta_2 = & 2\lambda_2 - 2N_C\lambda^2_2 + 2N_f\lambda_2\lambda_1 + 6\lambda_2\lambda_3 + 2\lambda_2\lambda_4 -  \\ 
        & - 12C_2(F)\lambda_2\alpha + 12\lambda_3\alpha - \frac{3}{2}\left(3N_C-\frac{4}{N_C}-\frac{1}{N_C^2}\right) \alpha^2,\\
        \beta_3 = & 2\lambda_3 - \frac{1}{4}\lambda^2_2 -\lambda_2\lambda_1 - 3\lambda^2_3 - N_f\lambda_2\lambda_4 + 2(N_C+N_f)\lambda_1\lambda_3 +\\
        & + 2(N_CN_f+1)\lambda_3\lambda_4 + \frac{6}{N_C}\lambda_3\alpha + \frac{3}{4}\left(1+\frac{3}{N_C^2}\right) \alpha^2,\\
        \beta_4 = & 2\lambda_4 - 3\lambda^2_1  - N_CN_f\lambda^2_3 + (N_CN_f -2)\lambda_4^2 - N_f\lambda_2\lambda_3 + \\
        & + 2(N_CN_f+1)\lambda_1\lambda_4 + 6(\lambda_1+\lambda_4)\alpha - \frac{3}{4}\left(3+\frac{1}{N_C^2}\right) \alpha^2.
    \end{split}
\end{align}

The above equations together with the equation for the gauge coupling lead to a 5-dimensional system that is valid for every value of $N_C$ and $N_f$. We are interested in the Veneziano limit, i.e. large $N_C$ and $N_f$ with fixed ratio $r$. First we rescale the couplings in the form 
\begin{align}
    \begin{split}
        N_C\lambda_{1,2} &\rightarrow\lambda_{1,2}, \\
        N_C\alpha &\rightarrow \alpha, \\
        N_C^2\lambda_{3,4} &\rightarrow\lambda_{3,4} .
    \end{split}
\end{align}
Then the equations for $\beta_1$ becomes
\begin{align}\label{beta1rg}
\begin{split}
    \frac{\beta_1}{N_C} = & \frac{1}{N_C}[ 2\lambda_1 + \frac{r}{4}\lambda^2_2+ (1+r)\lambda^2_1 - \frac{6}{N_C^2}\lambda_1\lambda_4 -\\
        &  - \frac{6}{N_C}(\lambda_1+\frac{1}{N_C^2}\lambda_4)\alpha - \frac{3}{4N_C}\left(N_C-\frac{8}{N_C}+\frac{3}{N_C^2}\right) \alpha^2] \\
    \beta_1 = & 2\lambda_1 + \frac{r}{4}\lambda^2_2+ (1+r)\lambda^2_1 - \frac{3}{4}\alpha^2.
\end{split}    
\end{align}
Similarly for $\beta_2$ and $\alpha$ we got

\begin{equation}\label{beta2rg}
    \beta_2 = 2\lambda_2-2\lambda_2^2+2r\lambda_1\lambda_2-6\alpha\lambda_2-\frac{9}{2}\alpha^2,
\end{equation}
\begin{equation}\label{alphargcara}
    \beta_\alpha = -\frac{2}{3}(11-2r)\alpha^2 - \frac{2}{3}(34-13r)\alpha^3 + 2r\alpha^2\lambda_1.
\end{equation}
where the last term in (\ref{alphargcara}) is a 4-fermi non-perturbative correction to the RG flow of the gauge coupling. Note that the system decouples from $\lambda_3$ and $\lambda_4$, so in order to study the structure of the theory space we only need to consider the 3-dimensional system. 

\begin{figure}
    \centering
    ~ 
    \begin{subfigure}[b]{0.6\textwidth}
        \includegraphics[width=\textwidth]{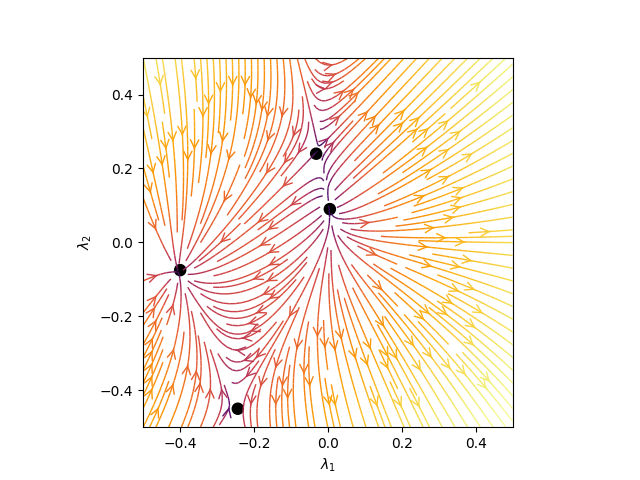}
        \caption{$r=4.1$}
    \end{subfigure}
    \begin{subfigure}[b]{0.6\textwidth}
        \includegraphics[width=\textwidth]{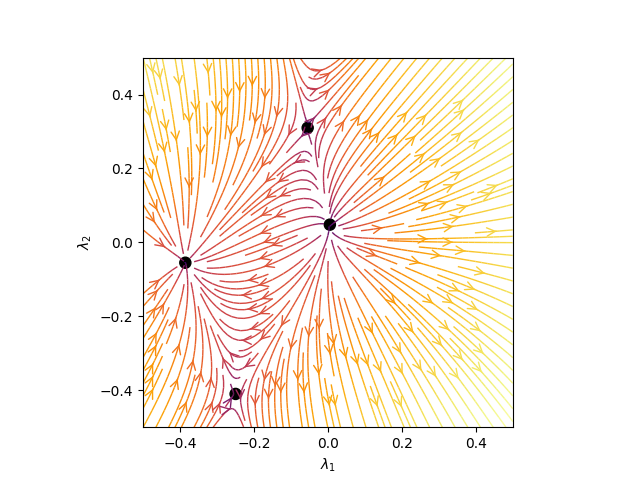}
        \caption{$r=4.3$}
    \end{subfigure}
    \begin{subfigure}[b]{0.6\textwidth}
        \includegraphics[width=\textwidth]{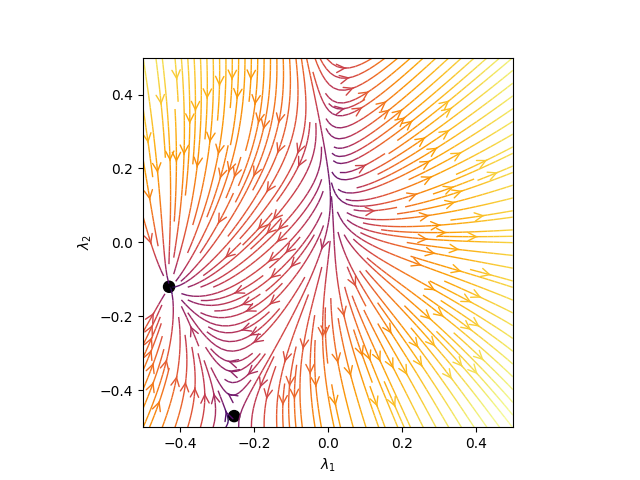}
        \caption{$r=3.9$}
    \end{subfigure}    
    \caption{$(\lambda_1,\lambda_2)$ plane with $\alpha=\alpha^*$. Black dots represent fixed points. Note how two fixed point merge and disappear below a critical value for r.}
    \label{QUECACEadcxzcfrqwes}
\end{figure}

There are two solutions for $\beta_\alpha = 0$: the trivial point $\alpha = 0$ and $\alpha^* = \frac{11-2r-3r\lambda_1}{13r-34}$. Using the linearization (with really small $\lambda_1$) of this equation, we can see that $\alpha^*$ is stable for $34/13<r<11/2$. We call this range of values for $r$ as the conformal window, since for this range a conformal color interacting theory is RG stable.

In order to look for fixed points in the flow we only need to input the two values that $\beta_\alpha$ vanishes in (\ref{beta1rg}) and (\ref{beta2rg}). For this, we will have two quadratic equations, so for each of the two cases we can expect a maximum of 4 fixed points in the $(\lambda_1,\lambda_2)$ plane. Numerically, \cite{QCDSADLE} shows that for $\alpha^*$, there is a saddle-node bifurcation in the plane for $r_{crit} \approx 4.05$. This shows color interacting fixed points disappearing below a critical value of the ratio. See figure \ref{QUECACEadcxzcfrqwes}.



\section{Holographic RG Flows}

An interesting set of theories are the ones that follow the holographic principle\cite{Holprinc}. The most well known example of such principle is the AdS/CFT correspondence\cite{ADSCFT}, where it is conjectured that there are relations between observables in conformal field theories (CFT's) and observables in gravity theories that live on the boundary of an anti-de Sitter space. 

We can build an AdS space by starting from the Einstein's equation:

\begin{equation}
    R_{\mu\nu} - \frac{1}{2}Rg_{\mu\nu} + \Lambda g_{\mu\nu} = 8\pi G T_{\mu\nu},
\end{equation}
where $\Lambda$ is called the cosmological constant and $G$ is the gravitational constant. We can find maximally symmetric solutions to this equations by looking at solution in the vacuum $T_{\mu\nu}=0$. Then we will have

\begin{equation}
    R = \frac{2\Lambda d}{d-2}.
\end{equation}
The solution for $\Lambda = 0$ is the Minkowski space $\mathcal{M}^d$, for positive $\Lambda$ is called de Sitter space and for negative $\Lambda$ the solution is called anti-de Sitter space AdS. 

An AdS$_{(d+1)}$ can be embedded in an higher dimensional Minkowski space $\mathcal{M}_{(d+2)}$ with the following metric:
\begin{equation}
    ds^2 = -(dX^0)^2 + (dX^1)^2 + (dX^2)^2 + ... +(dX^d)^2 - (dX^{d+1})^2 .
\end{equation}
The AdS$_{(d+1)}$ space is the given by the set:
\begin{equation}
    AdS_{(d+1)} = \left\lbrace X \in \mathcal{M}_{(d+2)} | -(X^0)^2 + \sum_{i=1}^d(X^i)^2 - (X^{d+1})^2 = L^2 \right\rbrace,
\end{equation}
where $L$ is called the AdS radius and we use the notation $X=(X^0,...,X^{d+1})$.

A more useful parametrization of such spaces is known as the Poincar\'e patch coordinates. We start by using the transformations from coordinates $(X^0,...,X^{d+1}$ to coordinates $(t,x^i,r)$, where $i=(1,...,d-1)$, given by:
\begin{align}
\begin{split}
        X^0 = & \frac{L^2}{2r}\left( 1 + \frac{r^2}{L^4}((x^i)^2 - t^2 + L^2 ) \right),\\
        X^i = & \frac{rx^i}{L},\\
        X^d = & \frac{L^2}{2r}\left( 1 + \frac{r^2}{L^4}((x^i)^2 - t^2 - L^2) \right),\\
        X^{d+1} = & \frac{rt}{L}.
\end{split}
\end{align}
Then the metric of the AdS space becomes:
\begin{equation}
    ds^2=\frac{L^2}{r^2}dr^2 + \frac{r^2}{L^2}(\eta_{\mu\nu}dx^\mu dx^\nu)
\end{equation}
We can still do another reparametrization if we define $r' = L\ln \frac{r}{L}$ then the metric is tranformed to (dropping the primes):

\begin{equation}
    ds^2 = e^{\frac{2r}{L}}dx^\mu dx^\nu \eta_{\mu\nu} + dr^2.
\end{equation}

As for the gravity theory that lives in this space we consider the d+1 dimensional action:

\begin{equation}
    S = d^xdr\sqrt{-g}\left( \frac{R}{16\pi G} - \frac{1}{2}\partial_m\phi\partial^m\phi-V(\phi) \right).
\end{equation}
The equations of motion are given by the equation from the dilaton

\begin{equation}
    \frac{1}{\sqrt{-g}}\partial_m(\sqrt{-g}g^{mn}\partial_n\phi) - V'(\phi) = 0,
\end{equation}
and the equation for the metric
\begin{equation}
    R_{mn} - \frac{R}{2}g_{mn} = 8\pi G (\partial_m\phi\partial_n\phi - \frac{1}{2}\partial_l\partial^l\phi - g_{mn}V(\phi)).
\end{equation}

A detailed calculation of these equations of motions can be seen at \cite{Naga}. Setting $\phi$ as constant gives us $R_{mn} - \frac{R}{2}g_{mn} = 8\pi Gg_{mn}V(\phi)$, which means that, if we set $8\pi GV(\phi)=\Lambda$, we see that the equation recovers to the Einstein equation for an AdS space.

Now, one could analyze perturbations around the AdS space and look at how the correlated theories that live on the boundary respond. This is aimed as an attempt to generalize the AdS/CFT correspondence. Since CFT's are fixed point in RG flows, we can expect that the starting point of a 'Holographic RG Flow' is an AdS space. Such hypothetical relations are known as Gauge/Gravity duality. We can deform the AdS space using the \textit{Domain Wall} metric as a general ansatz. The following part of this subsection is inspired by chapter 9 of \cite{JOJO}. The Domain Wall metric is given by:

\begin{equation}
    ds^2 = e^{2A(r)}\eta_{\mu\nu}dx^\mu dx^\nu + dr^2, \hspace{10 pt} \phi=\phi(r).
\end{equation}
Here $r$ defines a scale for the theory just like the energy scale $\mu$ defines the scale in field theories. In these sense, since we expect AdS spaces to be the critical solutions of the 'holographic flow', we can set at the limit $r\rightarrow \pm\ \infty $ a linear function $A=r/L$ and constant fields $\phi$. In fact, the relation between the scales is even deeper, we can write $\mu=\mu_0e^{\frac{r}{L}}$, and we can see that we find UV (IR) theories in the limit $\infty$ ($-\infty$).

For the ansatz above, we can calculate the Einstein tensors and we get
\begin{align}
    \begin{split}
        G^\mu_\nu = & (d-1)\delta^\mu_\nu\left(A''+\frac{d}{2}(A')^2\right)=8\pi G T^\mu_\nu, \\
        G^r_r = & \frac{d(d-1)}{2}(A')^2 = 8\pi G T^r_r,
    \end{split}
\end{align}
where we used the notation $A'=\partial_rA$. We can find an inequality by considering the equation for $G^t_t-G^r_r$
\begin{equation}
    A''= \frac{8\pi G}{d-1}(T^t_t-T^r_r) = \frac{8\pi G}{d-1}\phi .
\end{equation}
This means that 
\begin{equation}\label{popaocsp}
    A''\leq 0.
\end{equation}

The equations of motion are
\begin{align}
    \begin{split}
        \label{HlRg}
        \phi '' = & -dA'\phi' + \frac{dV(\phi)}{d\phi}, \\
        (\phi ')^2  = & 2V(\phi) + \frac{1}{8\pi G}d(d-1)(A')^2.
    \end{split}
\end{align}

This set of equations are called as the \textit{Holographic RG Flow equations}. The fact that it is of second order in $\phi$ makes it a little inconvenient. We can bypass that by defining an auxiliary function to reduce the system to differential equations of first order. One way we can do it is by defining the superpotential $W(\phi)$ by the equation

\begin{equation}\label{spuerpaopasof}
    V= \frac{1}{2}\left(\frac{dW}{dr}\right)^2 - \frac{d}{d-1}W^2,
\end{equation}
and then the system becomes:
\begin{align}
    \begin{split}
        \frac{dW}{d\phi} = & \sqrt{8\pi G}\frac{d\phi}{dr}, \\
        A'  = & -\frac{\sqrt{8\pi G}}{d-1}W.
    \end{split}
\end{align}
The system described above is a gradient system, hence we do not need to be worried about periodic flows in the solutions of the equations. 

We also can define a \textit{Holographic C-funcion}:

\begin{equation}
    C(r) = \frac{\pi}{G_5A'(r)^3},
\end{equation}
where 
\begin{equation}
G_5=\frac{\pi L^3}{2N^2}
\end{equation}
is Newton's constant in 5 dimensions. From \ref{popaocsp} we have that
\begin{equation}
    C'(r) = -3\frac{\pi}{G_5A'(r)^4}A''(r) \geq 0,
\end{equation}
and this means that $C(r)$ is always decreasing as the flow goes to the IR $(r\rightarrow\infty)$.

\section{Holographic approach to QCD}

In \cite{ImproQCD} and \cite{ImproQCD2} the authors aim to develop a holographic model that is related to QCD via the AdS/CFT correspondence. Here we are interested in looking at the holographic RG flows in these theories and we are going to see if the behaviours that appear in QCD also happen in the models studied by \cite{ImproQCD}.

The theories of interest here are four-dimensional $U(N_C)$ gauge theories at large $N_C$ and so, with no extra adjoint fields, the gravity dual theory is five-dimensional. The starting action for the string theory that we are going to use is:

\begin{equation}\label{ehessaaquimesmo}
    S_5 = M^3\int d^5x \sqrt{g}\left[e^{-2\phi}\left(R+4(\partial\phi)^2+\frac{\delta c}{\ell_s^2}\right)-\frac{1}{2\dot 5!}F_5^2-\frac{1}{2}F_1^2 -\frac{N_f}{\ell_s^2}e^{-\phi}\right],
\end{equation}

Here we are going to give just a general idea of what the terms means. The fields of the theory are simply a metric $g_{\mu\nu}$ and a dilaton $\phi$. $F_0$, $F_1$, $F_2$, $F_3$, $F_4$ and $F_5$ are called \textit{Ramond-Ramond fields} that appear in type II supergravity \cite{Ramond}-\cite{Ramond2}. We have that $F_0\sim F_5$, $F_1\sim F_4$ and $F_2\sim F_3$, so only $F_5$, $F_1$ and $F_2$ are independent fields.

$F_5$ generates a four form that comes from the branes responsible for the $U(N_C)$ group and its dual in Yang-Mills theory is a zero-form field strength. $F_1$ genereates an axion field $F_1 = \del_\mu a$ and its dual is $Tr[F\wedge F]$. $F_2$ generates a vector but it has no candidate for a dual in Yang-Mills theories \cite{ImproQCD}.

The last term in (\ref{ehessaaquimesmo}) refers to flavor coming from $N_f$ $D_4-\bar{D}_4$ brane pairs. It is also established that $\delta c = 5$ and that

\begin{equation}
    M^3 = \frac{1}{g_s^2\ell_s^3},
\end{equation}
where $g_s$ is the Plank length and $\ell_s$ is the string scale. Now we define
\begin{equation}
    \lambda = N_C e^\phi,
\end{equation}
and the Einstein frame $g_{\mu\nu} = \lambda^{4/3}g_{\mu\nu}^E$. The equations of motion for $F_5$ gives

\begin{equation}
    F_{1,2,3,4,5} = \frac{N_C}{\ell_s}\lambda^{10/3}\frac{\epsilon_{1,2,3,4,5}}{\sqrt{-\text{g}}}.
\end{equation}
Plugging this back into the action gives \cite{ImproQCD}

\begin{equation}
    S_5 = M^3N_C^2\int d^5x \sqrt{g}\left[R+\frac{(\partial\lambda)^2}{\lambda^2}-\frac{\lambda^2}{2}(N_C^2\del a)^2+ V(\lambda)\right],
\end{equation}
where the potential is
\begin{equation}\label{potentica}
    V(\lambda) = \frac{\lambda^{4/3}}{\ell_s^2}\left[ \delta c -x\lambda -\frac{1}{2}\lambda_2 \right].
\end{equation}
Here, $x$ is the ration of number of flavor and number of colors like we have defined in (\ref{oipo}),
\begin{equation}
    x=\frac{N_f}{N_C}.
\end{equation}

\subsubsection{Phase space variable}

In 5 dimensions the Holographic RG equations are (using the definitions on \cite{ImproQCD} this time)
\begin{align}
    \begin{split}
        \phi '' = & -4A'\phi' - \frac{3}{8}\frac{dV(\phi)}{d\phi}, \\
        (\phi ')^2  = & -\frac{9}{4}A'', \\
        A'^2 = & \frac{(\phi ')^2}{9} + \frac{V}{12}.
    \end{split}
\end{align}

We rewrite the fields as
\begin{equation}
    \Phi = \phi + \log N_C
\end{equation}

Another way we can simplify the system above is by defining a \textit{phase space variable} $X$ 
\begin{equation}
X=\frac{\phi '}{3A'}. 
\end{equation}
With the use of $X$, the new equations are \cite{ImproQCD}:

\begin{align}\label{holqcdeq}
    \begin{split}
        (\Phi')^2 = & \dfrac{3}{4}V_0X^2e^{\frac{8}{3}\int^\Phi_{-\infty}Xd\Phi},\\
        (A')^2 = &  \dfrac{V_0}{12}e^{-\frac{8}{3}\int^\Phi_{-\infty}Xd\Phi},\\
        N_C^2\frac{dX}{d\Phi} = & \left(8X + 3\frac{d\log V}{d\Phi}\right)\dfrac{X^2-1}{6X},
    \end{split}
\end{align}
where $V_0 >0$ is the limit of the dilaton potential at $\phi\rightarrow -\infty$.

There is a neat relation between the phase space variable and the superpotential in (\ref{spuerpaopasof}),

\begin{equation}
    X = -\frac{3}{4}\frac{d\log W}{d\phi}.
\end{equation}

We can also relate $X$ with the beta function of the coupling,

\begin{equation}
    \beta = \mu \frac{d\lambda}{d\mu}.
\end{equation}

An argument can be made that  \cite{ImproQCD}

\begin{equation}
    A \propto \log E.
\end{equation}
Then we can say that

\begin{equation}
    X = \frac{\beta}{3\lambda}
\end{equation}

\subsection{Single exponential potential}
In appendix D of \cite{ImproQCD}, the authors explore a potential with a single exponential 

\begin{equation}
    V(\phi) = \frac{4}{3}\epsilon^{8/(3a) \hspace{2pt}\phi}, \hspace{15 pt} \epsilon = \pm 1.
\end{equation}
This potential is the limit of $\phi\rightarrow\infty$ of (\ref{potentica}). Inputting this in (\ref{holqcdeq}) and calculating $\frac{d\log V}{d\phi} = \frac{8}{3a}$ we get

\begin{equation}
    \frac{dX}{d\phi} = F(X) = 8(X+\frac{1}{a})\frac{X^2-1}{6X} = \frac{4}{3}(X+\frac{1}{a})\frac{X^2-1}{X}.
\end{equation}

There are three fixed points, namely at $X=\pm 1$ and at $X=-1/a$. For simplicity, let us consider only the case that the parameter $a$ is positive. The stability analysis around the fixed point gives us that
\begin{equation}
    F'(X) = \frac{4}{3a}\frac{X^2-1}{X} + \frac{4}{3}(X+\frac{1}{a})\frac{2X-1}{X} - \frac{4}{3}(X+\frac{1}{a})\frac{2X-1}{X^2}.
\end{equation}

For $X=1$ $F'$ is always positive. Then we can say that the fixed point there is unstable. For the fixed point at $X=-1/a$, $F'(-1/a)=-\frac{4}{3}(\frac{1}{a^2}-1)$, then we can see that for $a>1$ the fixed point is unstable and for $a<1$ the fixed point is stable and the reverse happens for the fixed point at $X=-1$. This description of the exchange of stability at $1/a=1$ exactly matches the description of a transcritical bifurcation, the same bifurcation that occurs for the gauge coupling at the 4-dimensional QCD at two loops. For negative $a$ we would have the same bifurcation between the fixed points $X=1$ and $X=-1/a$. See figure \ref{fuiasofa}.

\begin{figure}[h]
\centering
\includegraphics[width=0.8\linewidth]{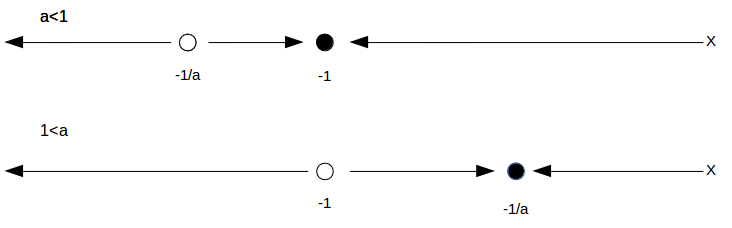}
\caption{Diagram showing the exchange of stability between the fixed points at $a=1$. Open circles represent unstable fixed points and closed circles represent stable fixed points. For $a<1$ the fixed point at $-1/a$ is unstable and the fixed point at $-1$ is stable. They exchange their stability properties for $a>1$.}
\label{fuiasofa}
\end{figure}

We also can analyze the fixed point solutions more closely using the other equations in (\ref{holqcdeq}). For $X=1$ we have to solve
\begin{align}
    \begin{split}
        \phi'' = -\frac{4}{3}\phi' = 0,\\
        \phi' = & 3A'.
    \end{split}
\end{align}
The solution is 
\begin{equation}
    \phi = \log c(r_0-r)^{3/4},
\end{equation}
where $c$ is a constant and $r_0$ is a length scale. Similarly for $X=-1$ we get the solution
\begin{equation}
    \phi = - \log c(r_0-r)^{3/4},
\end{equation}
and for $X=-1/a$
\begin{equation}
    \phi = \log \frac{c}{(r_0-r)}^{-3a/4}.
\end{equation}

The ideas laid down in this section can be used to explore more interesting potentials, like in \cite{Exotic} and \cite{Exotic2}. This certainly could be the topic for further research.



\section{Efimov Physics}

Non-relativistic field theories can also give interesting effects to analyse in RG flows. 

In particular, Efimov \cite{OGefi} found in 1970 that a model of three interacting non-relativistic bosons has a predicted series of a three-body energy levels when at least two pairs of particles has a two-body interaction form states that are a dissociation threshold. This means that we can have a bound state even if the two-body interaction is not strong enough to form bond between two particles. If one particle is removed the entire bound state dissipates. This states are called \textit{Efimov states}, and the effect is known as the \textit{Efimov effect}. 

The effect happens when the scattering length $a$ of a pair of bosons goes to infinity $a\rightarrow \infty.$\footnote{Or equivalently the binding energy goes to zero} This limit is known as unitary limit. Then the three body system shows the following scattering lengths $a_n$ and energies $E_n$ ($n\in \mathcal{N})$:
\begin{align}
    a_n = & a_0\lambda^n, \\
    E_n = & E_0\lambda^{-2n},
\end{align}
where $\lambda \approx 22.7$ it is known as \textit{Efimov's scaling constant}. The constant is related to the order $s_0$ of the imaginary-order modified Bessel function of the second kind that describes the radial dependence of the state by 
\begin{equation}
\lambda = e^{\pi s_0}, \hspace{20pt} s_0 \approx 1.00624.
\end{equation}

One could try to look for the RG of this system. In \cite{ConfLoss} the authors make an argument that when changing the dimension $d$ that the theory lives in, two fixed points collide in a certain critical value of the dimension (not necessarily an integer) and a limit cycle appears. It is important to say that, as argumented in \cite{efi}, limit cycles in such theories happens when the RG cutoff $\Lambda$ of the theory is changed by a specific multiplicative factor $\lambda_0$. This means that the theories in the cycle have a \textit{discrete scale invariance} with \textit{scaling factor} $\lambda_0$. 

In reference \cite{efi}, the authors makes the assumption of existence of a limit cycle more explicit. They see that there are actually two critical dimensions, $d_1$ and $d_2$, where bifurcations exists. For $d<d_1$ and $d_2< d$ the $\beta$-function has two fixed points and they correspond to two fixed points in the theory space. But for $d_1<d<d_2$ the fixed points are no present, and a limit cycle occurs. This description matches exactly the description of a saddle node infinite period (SNIPER) bifurcation that we saw in section \ref{globalglovalgoaskock}.

\newpage


\chapter{Conclusions}

The physics of a system depends on the energy scale used. This dependence generates equations that form dynamical systems, which are partial differential equations. By shifting the focus from trying to find analytical solutions of those equations to looking for more qualitative aspects of the solutions, we can find many effects that can occur in their behaviour with minimal data available to us.

One major thing of interest we analyse is the possible asymptotic behaviours of these solutions. Traditionally, those are considered to be flow lines between two fixed points. In field theory terms, the flow lines represent how a conformal field theory can be deformed via a relevant operator into another conformal field theory. The existence of a decreasing potential function (collectively called C-function) points out that these flows have a natural direction, from a UV CFT to an IR CFT. The connection between this function and fundamental properties of the theories let us interpret that the flow line can only move in the direction of the loss of degrees of freedom in a theory.

Other possible asymptotic behaviours are however possible. In 1971, Kenneth Wilson himself \cite{1Wilson71} described that, for systems with more than one coupling constant, it was possible that a periodic flow (also known as limit cycle) may appear. The proof of Zamolodchikov's C-theorem seemed to rule out that possibility for quantum field theories in two dimensions, and later conjectures for higher dimensions were prove in the following decades.

In 2012 Fortin, Grinstein, and Stergiou \cite{FGCDec} discovered examples of four dimensional field theories that presented limit cycles. At first, the authors thought this was an example of field theories that were scale invariant but not conformal invariant. Later, thanks to contributions of Osborn, Luty, Polchinski and Rattazzi \cite{FGCNov}, \cite{LPR}, \cite{JO2000}, they understood that this was not the case, but actually that these limit cycles represented ambiguities from the Callan-Symanzik equations with flavor symmetries. In this thesis we showed that their findings were not conflicting with our previous notion of RG flows if we consider the C-function to be constant around the cycle. 

Recently, there has been an interest in using more advanced techniques from dynamical systems in the study of RG flows \cite{GukovCount}, \cite{GukovBif}. We gave a short historical overview of this mathematical field and introduced basic theorems that hopefully are the minimum amount necessary for the reader to grasp the notions of dynamical systems. 

An important development was the Morse-Bott theory \cite{MORSEBOTT}, in name of Marston Morse and Raoul Bott, that connects algebraic topology with the study of manifolds that comes from fixed points in gradient dynamical system. One of the main results from this theory is the development of the Conley-Index, named after Charles C. Conley, which is a number that indicates the structure of fixed points inside of a region in the space. In this work we reviewed the basics of algebraic topology in order to give a definition of the Conley-Index.

Bifurcation theory is also presented together with examples in field theories. Bifurcations can explain ``first order phase transition'' in the RG diagrams. This phase transitions might serve as an origin of apparent fine tuning that need to exist for marginal operators. We also introduced RG flows in holographic theories, and we suggested that one may look at bifurcations in both sides of the duality.

\newpage

\appendix

\chapter{Appendix: Morse Theory}

Morse Theory is a field of mathematics that studies the topology of a manifold $M$ by with a metric $g$ looking at differentiable functions $f:M\rightarrow\Re$ that live on the manifold \cite{MORSEBOTT2}. It has been used in quantum field theory to study RG flows \cite{GukovBif} and supersymmetry \cite{Wits}. A critical point of $f$ is a point $m$ in $M$ such that $\nabla_g f(m)=0$, where $\nabla_g$ is the gradient operator.

\begin{mydef}
Let $m$ be a critical point of f. Then the \textbf{Morse index} $\mu$ of $m$ is the number of negative eigenvalues of the Hessian of f,
$$
(H_{ij}) = \left(\left.\dfrac{\del^2 f}{\del x_i\del x_j}\right|_{m}\right).
$$
Furthermore, the fixed point is said to be \textbf{nondegenerate} if the kernel of $(H_{ij})$ is empty, i.e., it has no eigenvalues equal to zero.
\end{mydef}
The definition of the Morse index reflects the intuition of the index being the number of dimensions of the relevant manifold of the critical point.

\begin{mydef}
If every critical point of a function $f$ is nondegenerate, then f is a \textbf{Morse function}.
\end{mydef}

\begin{mydef}
A \textbf{chart} (or a \textbf{coordinate patch}) $\phi$ of a topological manifold $M$ is a homeomorphism between a open subset $U$ of $M$ to an open subset of an Euclidean space $V\subset \Re^n$,
$$\phi = (x_1,...,x_n) : U\rightarrow V$$
\end{mydef}

With these definitions we can state the important lemma:

\begin{lemma} \textbf{Morse Lemma}\\
Let $f:M\rightarrow\Re$ be a Morse function and $m$ be a critical point of $f$. Then for a neighborhood $U$ of m such that there exists a coordinate patch $\phi = (x_1,...,x_n)$ such that $\phi (m) = 0$. Also, if the Morse index of $m$ is $\mu$, the we have that
$$f(x)-f(m) = -\sum_{i=1}^\mu x_i^2 + \sum_{i=\mu+1}^n x_i^2$$
\end{lemma}

\begin{theorem}
Let $f:M\rightarrow\Re$ be a Morse function with $M$ compact. Then $M$ has the homotopy type of a cell complex with an $\mu$-cell for each critical point of $f$ with index $\mu$.
\end{theorem}

Proofs of this theorem can be found at \cite{MORSEBOTT3} and \cite{MORSEBOTT4}. A \textit{cell complex} (or CW complex) is another way of breaking down a manifold in order to compute the homology groups of said manifold. They can be viewed as a generalization of a simplex that we explored at section \ref{dasdasxczczxvadfwewesdv86wrgi123jtr934tyhug}. The main takeaway from the theorem above is that is possible to calculate the homology groups of $M$. See the appendix of reference \cite{AlgebraicTopology} for a construction of the topology of CW complexes.

Now let us define $C^\mu$ as the number of critical points of index $\mu$ of $f$. Also we define $b_\mu$, known as the \textit{the $\mu$-th Betti number} of $M$, as the rank of the $\mu$-th homology group of $M$, $H_\mu (M)$. Now we define the \textit{Morse polynomial}

$$M(t) = \sum_{i=0}^n m_it^i,$$
and the \textit{Poincar\'e polynomial}
$$P(t) = \sum_{i=0}^n b_it^i.$$
In particular, for $t=-1$ the polynomial equal the Euler characteristic $P(-1)=\chi$. We can state the Morse Inequality as $M(t)\leq P(T)$, but a stronger result is possible.

\begin{theorem}\textbf{Morse Inequalities}
$$M(t)-P(t)=(1+t)Q(t),$$
for some polynomial $Q(t)\geq 0$.
\end{theorem}
The theorem shows a connection between the topology of a space $M$ with properties of the fluxes in $M$. For example, the case $t=-1$ we have that
$$M(-1) = P(-1) = \chi .$$

\subsubsection{Morse-Bott Theory}

A generalization of Morse theory can be constructed to include functions with non-isolated critical points. We say that a manifold $V$ is a nondegenerate critical manifold of $f$ is every point in $V$ is a critical point of $f$ and that the Hessian is nondegenerate at the normal direction of $V$. A \textit{Morse-Bott function} $f:M\rightarrow \Re$ is a function such that all of its critical points belongs to a nondegenerate critical manifold. 

We can generalise the Morse inequalities for Morse-Bott functions. Let us say that $f$ has a set of nondegenerate critical manifolds $\left\lbrace V_i \right\rbrace$ with index $\mu_i$. We denote $P^{V_i}$ as the Poincar\'e polynomial of the manifold $V_i$. Then we define:
$$
MB(t) = \sum_{i=0}^n P^{V_i}(t)t^{\mu_i}
$$
as the \textit{Morse-Bott Polynomial}. Now we have that
$$MB(t)-P(t)=(1+t)Q(t),$$
for some polynomial $Q(t)\geq 0$.

\subsubsection{Morse Homology}

Now we recover the definitions of irrelevant and relevant manifolds given in section \ref{dasdasdasdxzc}:
$$\mathcal{I}(m) = \left\lbrace x \in M | \lim_{t\rightarrow \infty} x(t) = m \right\rbrace,$$
$$\mathcal{R}(m) = \left\lbrace x \in M | \lim_{t\rightarrow - \infty} x(t) = m \right\rbrace.$$

An interesting fact is that the sets $\mathcal{I}(m)$ and $\mathcal{R}(m)$ are respectively homeomorphic to $\Re^{(n-\mu)}$ and $\Re^\mu$. 

If $f$ is a Morse function with the set $\left\lbrace m_i\right\rbrace$ of critical points such that $\mathcal{I}(m_i)$ is transverse of $\mathcal{I}(m_j)$ for all $i\neq j$, then the function is said to be a \textit{Morse-Smale function}. The pair $(f,g)$ is called a \textit{Morse-Smale pair}. The moduli space can also be defined
$$
\mathcal{M}(m,n)=\mathcal{R}(m)\cup\mathcal{I}(n) 
$$
such that $\text{dim } \mathcal{M}(m,n) = \mu(m) - \mu(n)$. 

We can use the theory to create the concept of \textit{Morse Homology}\cite{AnotherMorse}. We first denote $\text{Crit}_\mu (f)$ the set of critical points of index $\mu$ of $f$. We can now define the \textit{Morse Complex} by the chain complex $C_i$ that is generated by $\text{Crit}_\mu (f)$:
$$
C_i^{\text{Morse}}(f,g) = \underbrace{\mathbb{Z}\oplus\mathbb{Z}\oplus\cdots\oplus\mathbb{Z}}_{\text{Crit}_\mu (f)}.
$$
We can define a boundary map $\del^{\text{Morse}}:C_i\rightarrow C_{i-1}$:
$$
\del^{\text{Morse}} (m) = \sum_{n\in \text{Crit}_{\mu-1} (f)} (\text{number of points in }\mathcal{M}(m,n)) \cdot n.
$$
We have that $(\del^{\text{Morse}})^2=0$. This is enough for us to define the Morse homology $H_*^{\text{Morse}}(f,g)$ as the homology of the chain complex $(C_*^{\text{Morse}},\del^{\text{Morse}})$. It can be shown that Morse homology is isomorphic to singular homology.

\subsubsection{Morse-Conley-Floer Homology}

The generalization that includes non-gradient functions is called \textit{Morse-Conley-Floer Homology} \cite{FloerHomo}. In this theory, instead of the Morse-Smale pair $(f,g)$ we have the pair $(S,\beta)$, where $\beta$ is a flow on the manifold and $S$ is an isolated invariant set of $\beta$ (see section \ref{conasldkla}). In particular, there is an analogue Morse inequality for Morse-Conley-Floer flows, given a decomposition of $S$, $\left\lbrace S_i \right\rbrace$, we have
$$\sum_{i\in I}P_t(S_i\beta)-P_t(S,\beta)=(1+t)Q(t).$$

We can also define the homology for the pair 
$$HI_*(S,\beta).$$

One interesting property is that if we define a set $N$ as a neighborhood of $S$ and $L\subset \del N$ as the set of points where $\beta$ is flowing outside of $N$, then we have that 
$$HI_*(S,\beta)\cong H_*(N/L,{L}),$$
and we essentially recovered the definition of Conley Index. 

As seen in section \ref{conasldkla}, this homology is invariant under continuous perturbations of $S$ and of $\beta$. If $(S_0,\beta_0)$ can be continuous transformed into $(S_1,\beta_1)$ then
$$HI_*(S_0,\beta_0)\cong HI_*(S_1,\beta_1).$$

\end{document}